\documentclass[aps,prl,twocolumn,superscriptaddress,nofootinbib]{revtex4-2}
\usepackage{amsmath,amssymb,amsfonts}
\usepackage{graphicx}
\usepackage{hyperref}
\usepackage{xcolor}

\usepackage{amsthm}

\usepackage{fancyhdr}
\usepackage[calcwidth]{titlesec}
\usepackage[font={footnotesize},labelfont={bf}]{caption}
\usepackage{setspace}

\usepackage{extarrows}

\usepackage[most]{tcolorbox}
\newtcolorbox{myt}[2][]{%
  attach boxed title to top center
               = {yshift=-4pt},
  colback      = blue!5!white,
  colframe     = blue!75!black,
  halign       = flush left,
  fonttitle    = \bfseries\sffamily,
  colbacktitle = blue!65!black,
  title        = #2,#1,
  enhanced,
}
\newtcolorbox{myd}[2][]{%
  attach boxed title to top center
               = {yshift=-4pt},
  colback      = violet!5!white,
  colframe     = violet!75!black,
  halign       = flush left,
  fonttitle    = \bfseries\sffamily,
  colbacktitle = violet!65!black,
  title        = #2,#1,
  enhanced,
}
\newtcolorbox{mye}[2][]{%
  attach boxed title to top center
               = {yshift=-4pt},
  colback      = purple!5!white,
  colframe     = purple!75!black,
  halign       = flush left,
  fonttitle    = \bfseries\sffamily,
  colbacktitle = purple!65!black,
  title        = #2,#1,
  enhanced,
}

\newtcolorbox{myg}[2][]{%
  attach boxed title to top center
               = {yshift=-4pt},
  colback      = green!5!white,
  colframe     = green!75!black,
  halign       = flush left,
  fonttitle    = \bfseries\sffamily,
  colbacktitle = green!65!black,
  title        = #2,#1,
  enhanced,
}

\providecommand{\U}[1]{\protect\rule{.1in}{.1in}}

\usepackage{dsfont}

\usepackage{mathtools}

\def\cost{{\rm Cost}}

\def\1{\mathbf{1}}

\def\lose{\text{\tiny LOSE}}

\def\prob{{\rm Prob}}

\def\spec{{\rm spec}}

\def\pure{{\rm PURE}}
\def\pos{{\rm Pos}}

\def\cptp{{\rm CPTP}}

\newcommand{\epm}{\end{pmatrix}}
\newcommand{\bpm}{\begin{pmatrix}}

\renewcommand{\log}{{\operatorname{log}}}

\newcommand{\ebm}{\end{bmatrix}}
\newcommand{\bbm}{\begin{bmatrix}}

\def\bmyd{\begin{myd}{}
\begin{definition}}
\def\emyd{\end{definition}\end{myd}}

\def\bmyl{\begin{myg}{}
\begin{lemma}}
\def\emyl{\end{lemma}\end{myg}}

\def\bmyt{\begin{myt}{}
\begin{theorem}}
\def\emyt{\end{theorem}\end{myt}}

\def\bmyc{\begin{myg}{}
\begin{corollary}}
\def\emyc{\end{corollary}\end{myg}}

\def\>{\rangle}
\def\<{\langle}

\def\id{\mathsf{id}}

\def\mE{\mathcal{E}}

\def\mF{\mathcal{F}}
\def\mN{\mathcal{N}}

\def\mP{\mathcal{P}}

\def\mV{\mathcal{V}}

\def\mQ{\mathcal{Q}}

\renewcommand{\geq}{\geqslant}
\renewcommand{\leq}{\leqslant}

\newcommand{\ben}{\begin{enumerate}}
\newcommand{\een}{\end{enumerate}}

\theoremstyle{definition}
\newtheorem{theorem}{Theorem}
\theoremstyle{definition}
\newtheorem{corollary}{Corollary}
\theoremstyle{definition}
\newtheorem{lemma}{Lemma}

\theoremstyle{definition}
\newtheorem{definition}{Definition}

\usepackage{array}
\usepackage{amsmath}

\newcommand{\bea}{\begin{eqnarray}}
\newcommand{\eea}{\end{eqnarray}}
\newcommand{\be}{\begin{equation}}
\newcommand{\ee}{\end{equation}}
\newcommand{\ba}{\begin{equation}\begin{aligned}}
\newcommand{\ea}{\end{aligned}\end{equation}}

\newcommand{\bee}{\begin{enumerate}}
\newcommand{\eee}{\end{enumerate}}

\def\be{\begin{equation}}
\def\ee{\end{equation}}

\newcommand{\herm}{{\rm Herm}}

\newcommand{\G}{\mathbf{G}}

\newcommand{\ml}{\mathfrak{L}}
\newcommand{\mb}{\mathfrak{B}}
\newcommand{\md}{\mathfrak{D}}
\newcommand{\mf}{\mathfrak{F}}

\newcommand{\mfu}{\mathfrak{U}}

\newcommand{\mO}{\mathcal{O}}

\newcommand{\lr}{\rangle\langle}
\newcommand{\la}{\langle}
\newcommand{\ra}{\rangle}
\newcommand{\tr}{{\rm Tr}}

\newcommand{\ua}{\uparrow}
\newcommand{\da}{\downarrow}

\newcommand{\eps}{\varepsilon}

\newcommand{\mbb}[1]{\mathbb{#1}}

\newcommand{\eqdef}{\coloneqq}

\def\s{\mathbf{s}}
\def\p{\mathbf{p}}
\def\q{\mathbf{q}}

\def\x{\mathbf{x}}

\def\t{\mathbf{t}}
\def\u{\mathbf{u}}

\def\0{\mathbf{0}}

\def\tA{\tilde{A}}

\def\trho{\tilde{\rho}}

\def\tH{\tilde{H}}

\def\eff{{\rm eff}}

\newcommand{\GG}[1]{\rm \textcolor{red}{ #1}{\color{red} \to}}

\newcommand{\Gg}[1]{
  \text{\fcolorbox{black}{red!20}{\textcolor{black}{\scriptsize$#1$}}}{\color{red} \xrightarrow{\hspace{0.1cm}\;\hspace{0.1cm}}}
}

\begin{document}

\title{Convex Split Lemma without Inequalities}
\author{Gilad Gour}
\affiliation{Technion - Israel Institute of Technology, Faculty of Mathematics, Haifa 3200003, Israel}
	
	\date{\today}

\begin{abstract}	
	We introduce a refinement to the convex split lemma by replacing the max mutual information with the collision mutual information, transforming the inequality into an equality. This refinement yields tighter achievability bounds for quantum source coding tasks, including state merging and state splitting. Furthermore, we derive a universal upper bound on the smoothed max mutual information, where ``universal" signifies that the bound depends exclusively on R\'enyi entropies and is independent of the system's dimensions. This result has significant implications for quantum information processing, particularly in applications such as the reverse quantum Shannon theorem. 
	\end{abstract}
	
	\maketitle


\noindent\textbf{Introduction:} The single-shot approach to quantum information science, introduced in the seminal work~\cite{Renner2005} and extensively developed in books such as~\cite{Hayashi2006,Tomamichel2015,KW2024,Gour2025}, has fundamentally reshaped our understanding of quantum information processing (QIP) tasks. Unlike the traditional asymptotic framework, which assumes many independent and identically distributed (i.i.d.) copies of a system, the single-shot perspective analyzes quantum resources and tasks in settings where resources are finite, and such assumptions are inapplicable. This shift has led to significant advancements across numerous QIP tasks, including single-shot quantum state merging~\cite{Berta2009}, classical communication over quantum channels~\cite{WR2012}, and many others.

The single-shot framework relies on smoothed entropic quantities to rigorously characterize resource trade-offs, offering two key advantages:

\ben \item It eliminates the need to invoke the law of large numbers in proving fundamental results in quantum Shannon theory, often leading to more intuitive and transparent formulations, even in the asymptotic limit.
\item It provides a robust approach for scenarios where the i.i.d.~approximation fails, including finite-resource systems, distributed networks with a limited number of nodes, and near-term quantum devices operating in noisy or low-repetition regimes.
\een

Despite its success, obtaining exact analytical expressions for optimal single-shot rates in QIP tasks remains a formidable challenge. Consequently, significant effort has been devoted to deriving tight upper and lower bounds. Such bounds not only recover known asymptotic results but also quantify how limited resources impact quantum tasks, enabling more accurate performance guarantees for protocols such as quantum communication, cryptography, and entanglement distillation.

Motivated by the need for tighter bounds, this paper introduces two key advances in quantum information theory:

\ben \item An \emph{equality-based convex split lemma}.
\item A \emph{universal upper bound on the smoothed max mutual information}.
\een

The convex split lemma~\cite{ADJ2017}, a fundamental tool inspired by classical rejection sampling techniques, plays a central role in proving achievability results for quantum source coding. Here, we present an equality-based formulation that replaces the conventional max mutual information with the \emph{collision mutual information} (i.e., the sandwiched mutual information of order two). This refinement converts the original inequality into an exact relation, leading to significantly tighter achievability bounds for quantum source coding tasks such as quantum state merging~\cite{HOW2005,Berta2009} and state splitting~\cite{ADHW2009,BCR2011}. The use of collision mutual information enables a more precise characterization of trade-offs in these tasks, particularly in the single-shot regime, where resource optimization is critical.

Universal bounds are upper or lower bounds on smoothed entropic quantities that are dimension-independent and expressed in terms of additive entropy functions. They highlight the unifying role of additive functions in quantum information across both single-shot and asymptotic regimes. For example, universal bounds on the hypothesis testing divergence provide a simple and elegant proof of the quantum Stein’s lemma, illustrating how asymptotic state distinguishability naturally emerges from single-shot quantities.

In this paper, we establish a universal upper bound for the smoothed max-information, with significant implications for tasks such as the reverse quantum Shannon theorem~\cite{BDH+2014,BCR2011}, which characterizes the resource cost of simulating quantum communication. Deriving this bound requires overcoming technical challenges due to the intricate dependence of smoothed max-information on bipartite correlations.

By refining fundamental tools in quantum information theory, this work provides deeper insights into single-shot quantum resource trade-offs and their implications for quantum communication and coding. The full proofs and technical details supporting our results are presented in the appendices.\\

\noindent\textbf{Notations:} We use $A$, $B$, and $R$ to denote both quantum systems (or registers) and their corresponding Hilbert spaces. The set of density operators on $A$ is denoted $\md(A)$, and the set of completely positive trace-preserving (CPTP) maps from $A$ to $B$ is $\cptp(A \to B)$.
The trace distance serves as our primary metric, and we define the $\eps$-ball around $\rho\in\md(A)$ as
\be 
\mb^\eps(\rho) \eqdef \left\{\sigma\in\md(A) \;:\; \frac12 \|\rho - \sigma\|_1 \leq \eps \right\}\;. 
\ee
We say $\sigma$ is $\eps$-close to $\rho$ and write $\sigma \approx_\eps \rho$ if $\sigma\in\mb^\eps(\rho)$. Additionally, we use the purified distance
\be P(\rho,\sigma) \eqdef \sqrt{1 - F^2(\rho,\sigma)} \ee
where the fidelity is given by $F(\rho,\sigma) \eqdef \|\sqrt{\rho} \sqrt{\sigma}\|_1$.\\

\noindent\textbf{Divergences and Relative Entropies:}
In this paper we primarily work with the \emph{sandwiched R\'enyi relative entropy}.  The sandwiched R\'enyi relative entropy of order $\alpha \in [0,\infty]$ is defined for every $\rho, \sigma \in \md(A)$ as~\cite{MDS+2013,WWY2014,Matsumoto2018b,GT2020}\footnote{The sandwiched Rényi relative entropy is often denoted with a tilde above $D_\alpha$. However, since it is the only divergence considered in this paper, we omit the tilde for clarity and simplicity.}
\be\nonumber
D_{\alpha}(\rho\|\sigma)=
	\begin{cases}
	\substack{\frac1{\alpha-1}\log Q_\alpha(\rho\|\sigma)\\}&\substack{\text{if }\frac12\leq\alpha<1\text{ and }\rho\not\perp\sigma, \text{ or }\rho\ll\sigma\\}\\
	\substack{\frac1{\alpha-1}\log Q_{1-\alpha}(\sigma\|\rho)\\}&\substack{\text{if }0\leq \alpha<\frac12\text{ and }\rho\not\perp\sigma\\}\\
	\substack{\infty\\}&\substack{\text{otherwise}\\}
	\end{cases}
\ee
Here, $\rho \ll \sigma$ means that the support of $\rho$ is a subspace of the support of $\sigma$, and $\rho \not\perp \sigma$ means $\tr[\rho \sigma] \neq 0$. The quantity $Q_\alpha(\rho\|\sigma)$ is defined as
\be
Q_\alpha(\rho\|\sigma)\eqdef\tr\left(\sigma^{\frac{1-\alpha}{2\alpha}}\rho\sigma^{\frac{1-\alpha}{2\alpha}}\right)^\alpha\;.
\ee 
Additionally, we define the $\alpha$-mutual information for a bipartite state $\rho^{AB}$ as
\be
I_\alpha(A:B)_\rho\eqdef\min_{\sigma\in\md(B)}D_{\alpha}\left(\rho^{AB}\big\|\rho^A\otimes\sigma^B\right)\;.
\ee

The sandwiched R\'enyi relative entropy admits special cases for $\alpha=0$, $\alpha=1$, and $\alpha=\infty$, which are understood in terms of limits.
For $\alpha=0$, it corresponds to the min relative entropy $D_{\min}$, defined as:
\be 
D_{\min}(\rho\|\sigma)\eqdef-\log\tr\left[\sigma\Pi_\rho\right]\;, 
\ee
where $\Pi_\rho$ is the projection onto the support of $\rho$.
For $\alpha=1$, it reduces to the Umegaki relative entropy:
\be 
D(\rho\|\sigma)\eqdef\tr[\rho\log(\rho)]-\tr[\rho\log(\sigma)]\;.
\ee 
For $\alpha=\infty$, it gives the max relative entropy:
\be
D_{\max}(\rho\|\sigma)\eqdef\inf_{t\in\mbb{R}_+}\big\{\log (t)\;:\;t\sigma\geq\rho\big\}\;.
\ee


The sandwiched R\'enyi relative entropy of order $\alpha = 2$, often referred to as the \emph{collision relative entropy}~\cite{BG2014}, will also play a central role in several applications discussed in this paper. Its name is derived from the concept of collision entropy, which is closely tied to the \emph{collision probability} in probability theory. For all $\rho, \sigma \in \md(A)$, it is defined as
\be
D_2(\rho\|\sigma)\eqdef\log Q_2(\rho\|\sigma)
\ee
where
\ba
Q_2(\rho\|\sigma)&\eqdef\tr\left(\sigma^{-\frac14}\rho\sigma^{-\frac14}\right)^2=\tr\left[\rho\sigma^{-\frac12}\rho\sigma^{-\frac12}\right].
\ea
This quantity has a particularly simple form due to the quadratic dependence of $Q_2$ on $\rho$. In the Appendix we discuss additional properties of the collision relative entropy and show in particular that for every $\rho,\sigma\in\md(A)$:
\ba\label{a6}
&D_2(\rho\|\sigma)\geq\log\left(1+\|\rho-\sigma\|_1^2\right)\quad\text{and}\\
&D_{2}(\rho\|\sigma)\geq-\log\left(1-P^2(\rho,\sigma)\right)\;.
\ea
Note that the second relation can be expressed as
\be\label{pqp}
P^2(\rho,\sigma)\leq 1-\frac{1}{Q_2(\rho\|\sigma)}\;.
\ee

Another key quantum divergence with strong operational significance is the hypothesis testing divergence. For any $\eps\in[0,1)$ and $\rho,\sigma\in\md(A)$, it is defined as:
\be\label{htd}
D_{\min}^\eps(\rho\|\sigma)\eqdef-\log\min_{0\leq\Lambda\leq I}\Big\{\tr[\Lambda\sigma]\;:\;\tr[\rho\Lambda]\geq1-\eps\Big\}\;.
\ee
This divergence is fundamental in quantum information processing, underpinning tasks such as quantum hypothesis testing, error exponents, and one-shot channel coding.\\

\noindent\textbf{Equality-Based Convex Split Lemma:}
The convex split lemma is a fundamental tool in quantum information theory, often employed in the analysis of single-shot quantum protocols~\cite{ADJ2017,AJW2018,AJW2019,AJ2022,ABS+2023}.  It establishes an upper bound on how well a quantum state can be approximated by a uniform mixture of related states. We begin by reviewing the standard convex split lemma before presenting an enhanced version that strengthens this bound.

Consider the setup where $n \in \mathbb{N}$ represents the number of copies of a system $A$, denoted as $A^n \coloneqq (A_1, \ldots, A_n)$. Let $R$ be a reference system, and let $\rho\in\md(RA)$ and $\sigma\in\md(A)$. Define the family of states $\{\tau_x^{RA^n}\}_{x\in[n]}$ as
\be\label{rcyc}
\tau^{RA^n}_x\eqdef\rho^{RA_x}\otimes\sigma^{A_1}\otimes\cdots\otimes\sigma^{A_{x-1}}\otimes\sigma^{A_{x+1}}\otimes\cdots\otimes\sigma^{A_n}\;.
\ee
The convex split lemma considers the uniform mixture 
\be\label{tau}
\tau^{RA^n}\eqdef\frac1n\sum_{x\in[n]}\tau^{RA^n}_x\;,
\ee
and establishes the bound
\be\label{gmain0}
D\left(\tau^{RA^n}\big\|\rho^R\otimes\left(\sigma^{A}\right)^{\otimes n}\right)\leq\log\left(1+\frac{\mu_{\max}} n\right)\;,
\ee
where $\mu_{\max}\eqdef2^{D_{\max}(\rho^{RA}\|\rho^R\otimes\sigma^R)}-1$. 

Applying the bound $\ln(1+x) \leq x$ and the Pinsker inequality, we obtain
\be\label{pinsk}
\frac12\left\|\tau^{RA^n}-\omega^R\otimes\left(\sigma^{A}\right)^{\otimes n}\right\|_1\leq\sqrt{\frac{\mu_{\max}} {2n}}\;.
\ee
This shows that for large $n$, $\tau^{RA^n}$ approaches $\rho^R \otimes \left(\sigma^A\right)^{\otimes n}$. Moreover, combining~\eqref{gmain0} with the purified distance bound in~\eqref{pqp} yields 
\be\label{split7} 
P^2\left(\tau^{RA^n}, \rho^R \otimes \left(\sigma^A\right)^{\otimes n}\right) \leq \frac{\mu_{\max}}{n + \mu_{\max}}. 
\ee 
Thus, for $n \gg \mu_{\max}$, $\tau^{RA^n}$ is very close to $\rho^R \otimes \left(\sigma^A\right)^{\otimes n}$.

We now present an enhanced version of the convex split lemma that strengthens the bound in~\eqref{gmain0} by replacing the Umegaki relative entropy with the collision relative entropy, thereby upgrading the inequality to an exact relation.
Let $n\in\mbb{N}$, $\rho\in\md(RA)$, $\sigma\in\md(A)$, $\omega\in\md(R)$, and let $\tau^{RA^n}$ be defined as in Eqs.~(\ref{rcyc},\ref{tau}). Then:

\begin{myt}{\small \color{yellow} Equality-Based Convex-Split Lemma}
\begin{lemma}\label{split}
\ba\label{gmain}
&Q_2\left(\tau^{RA^n}\big\|\omega^R\otimes\left(\sigma^{A}\right)^{\otimes n}\right)
= \\
&\frac{n-1}nQ_{2}(\rho^{R}\|\omega^R)
+\frac1nQ_{2}(\rho^{RA}\|\omega^R\otimes\sigma^A)\;.
\ea
\end{lemma}
\end{myt}
\noindent{\it Remark.}
This result holds for any choice of $\omega^R$. In particular, setting $\omega^R = \rho^R$ simplifies the expression, yielding a form that closely resembles the original inequality~\eqref{gmain0}:
\be\label{gmain8}
D_2\left(\tau^{RA^n}\big\|\rho^R\otimes\left(\sigma^{A}\right)^{\otimes n}\right)=\log\left(1+\frac\mu n\right)
\ee
where $\mu\eqdef Q_{2}(\rho^{RA}\|\rho^R\otimes\sigma^A)-1$.
 Moreover, using the relation between trace distance and collision relative entropy from~\eqref{a6}, we obtain
\be
\frac12\left\|\tau^{RA^n}-\rho^R\otimes\left(\sigma^{A}\right)^{\otimes n}\right\|_1\leq\frac14\sqrt{\frac\mu n}\;.
\ee
This bound is tighter than~\eqref{pinsk} for two reasons: (i) $\mu < \mu_{\max}$, yielding a sharper approximation, and (ii) the prefactor $1/4$ is smaller than $1/\sqrt{2}$, improving the bound’s scaling.

In the following corollary, we define
\be\label{nu} 
\nu_n \eqdef \min_{\omega\in\md(R)} \left\{ \frac{n-1}{n} Q_{2}(\rho^{R} \|\omega^R) + \frac{1}{n} Q_{2}(\rho^{RA} \|\omega^R \otimes \sigma^A) \right\} 
\ee
Setting $\omega^R = \rho^R$ yields the upper bound
\be\label{nu2} \nu_n \leq 1 + \frac{\mu}{n}.
\ee
Following the same notations as in Theorem~\ref{split} we have:
\bmyc\label{cor1}
\be\label{split9} 
P^2\left(\tau^{RA^n}, \rho^R \otimes \left(\sigma^A\right)^{\otimes n}\right) \leq 1-\frac{1}{ \nu_n}\;. 
\ee 
\emyc
\noindent\textbf{Remark.}
Combining the corollary with~\eqref{nu2} results with
\be\label{pmu0}
P^2\left(\tau^{RA^n}, \rho^R \otimes \left(\sigma^A\right)^{\otimes n}\right) \leq\frac{\mu}{\mu+n}\;.
\ee

\noindent\textbf{Application - Quantum State Splitting:}
Quantum state splitting (QSS) is a source coding protocol that serves as the reverse or dual of quantum state merging (QSM). In QSM, Alice and Bob initially share a bipartite state $\rho^{AB}$, and the task is to transfer Alice's share, $A$, to Bob. In QSS, the process is reversed: consider an i.i.d. source producing composite pure states $\rho^{AA'}$, where both subsystems $A$ and $A'$ are held by Alice. The goal of QSS is to transfer the system $A'$ to Bob.
To formalize, let $B$ denote Bob's replica of $A'$. The objective of QSS is to simulate the action of the identity channel $\id^{A' \to B}$ on the purification $\rho^{RAA'}$ of the state $\rho^{AA'}$, where $R$ is the reference system (see Fig.~\ref{qss0}). 

\begin{figure}[h]\centering    \includegraphics[width=0.4\textwidth]{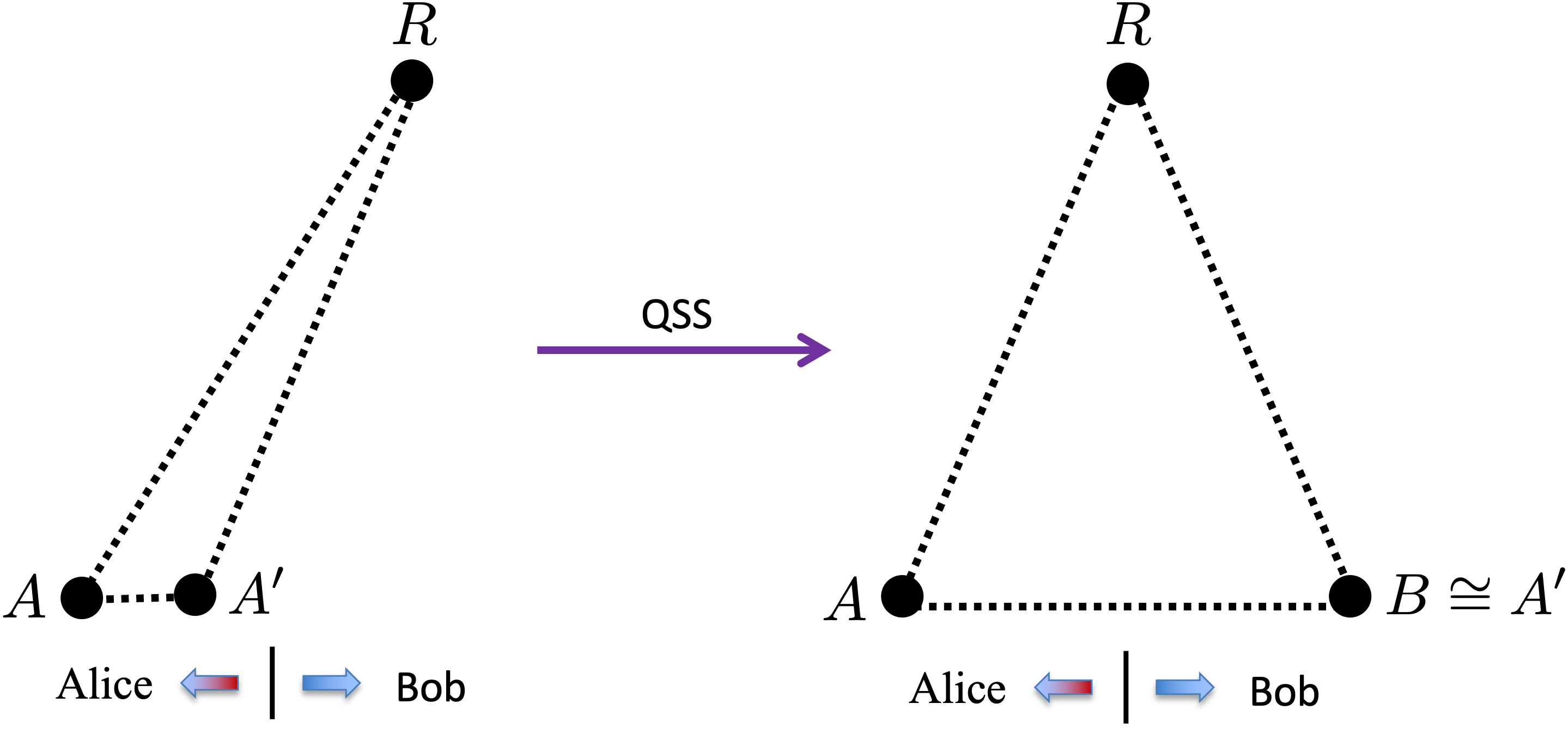}
  \caption{\linespread{1}\selectfont{\small Heuristic description of quantum state splitting.}}
  \label{qss0}
\end{figure}

Consider a setting where shared entanglement is freely available, but both classical and quantum communication are restricted. The free operations in this framework are therefore limited to local operations and shared entanglement (LOSE). The quantum communication cost of QSS under LOSE is the minimum number of qubits that Alice must transmit to Bob to approximate the action of $\id^{A' \to B}$ on $\rho^{RAA'}$ within an error tolerance of $\eps \in (0,1)$. We denote this cost by $\cost^\eps(\rho^{AA'})$. Figure~\ref{lose0} illustrates the action of such an LOSE superchannel on a communication channel $\id_m$ (representing $\log(m)$ noiseless qubit channels). 
 
 \begin{figure}[h]\centering    \includegraphics[width=0.4\textwidth]{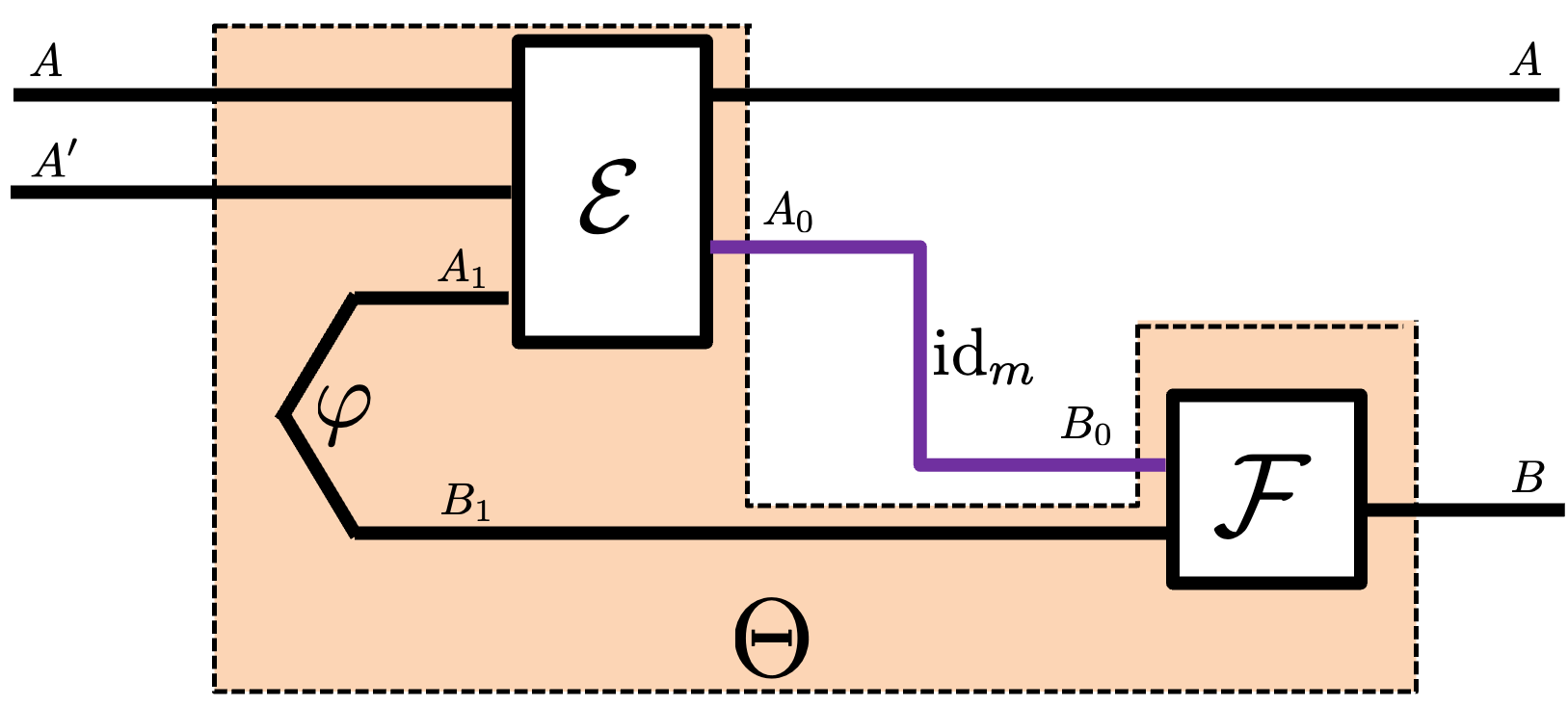}
  \caption{\footnotesize An LOSE  superchannel $\Theta$ (comprising of the channels $\mE$ and $\mF$, and the state $\varphi$) acting on a communication resource $\id_m$.}
  \label{lose0}
\end{figure}

In~\cite{BCR2011}, later refined in~\cite{ADJ2017}, an upper bound on the quantum communication cost was established:~\footnote{The smoothing in~\cite{BCR2011} and~\cite{ADJ2017} was performed using the purified distance, whereas we utilized the trace distance, necessitating some adjustments.}
\be
\cost^{2\eps}\left(\rho^{AA'}\right)\leq\frac12I_{\max}^\eps(R:A')_\rho+\log\left(\frac1\eps\right)\;
\ee
Here, we improve this bound by replacing the max mutual information with the smaller collision mutual information. Specifically, for the purification $\rho^{RAA'}$ of $\rho^{AA'}$ and any $0<\delta<\eps<1$, the communication cost for an $\eps$-error QSS under LOSE satisfies:
 \bmyt\label{thmeqss}
\be\label{ub00}
\cost^{\eps}\big(\rho^{AA'}\big)\leq\frac12I_{2}^{\eps-\delta}(R:A')_\rho+\log\frac1\delta\;.
\ee
\emyt

\noindent\textbf{Universal Bounds:}
Smoothed entropic functions play a central role in single-shot quantum information theory, characterizing optimal rates for various quantum tasks. These rates are typically bounded using additive functions, facilitating a smooth transition to the asymptotic limit.
In particular, the smoothed max-relative entropy $D_{\max}^\eps$ is fundamental in quantum information processing. Recently, it was shown in~\cite{RLD2025} that for any $\rho,\sigma\in\md(A)$, $\eps\in(0,1)$, and $\beta>1$,
\ba\label{8direct8}
D_{\max}^{\eps}(\rho\|\sigma)\leq D_\beta(\rho\|\sigma)+\frac1{\beta-1}\log\frac1{\eps^2}\;,
\ea
(see also Proposition 6.22 in~\cite{Tomamichel2015} for a similar bound under purified distance smoothing).
Notably, the bound is dimension-independent and additive under tensor products, with corrections depending only on $\eps$ and $\beta$.

We now consider the smoothed max mutual information, a central quantity in quantum information, defined as~\cite{BCR2011}
\be 
I_{\max}^\eps(A:B)_\rho=\min_{\trho\in\mb^\eps(\rho^{AB})}I_{\max}(A:B)_{\trho}\;. 
\ee 
Since $I_{\max}^\eps$ is not explicitly expressed in terms of $D_{\max}^\eps$, the bound in~\eqref{8direct8} does not directly apply, necessitating a different approach.

We derive a universal upper bound in the following theorem. For $\alpha,\eps\in(0,1)$ and $\beta>1$, let $c \eqdef 2-\sqrt{3} \approx 0.27$ and define
\be\label{0gag0} 
f_{\alpha,\beta}(\eps)\eqdef \left(\frac2{\beta-1}+\frac1{1-\alpha}\right)\log\left(\frac1{c\eps^2}\right)\;.
\ee
Then, for every $\rho\in\md(AB)$ we have:
\bmyt\label{uab}
\ba\label{final0} I_{\max}^\eps(A:B)_\rho&\leq H_\alpha(A)_{\rho}-\tH_{\beta}^{\ua}(A|B)_\rho+f_{\alpha,\beta}(\eps).
\ea 
\emyt

\noindent\textbf{Application - Channel Simulation under LOSE:}  Consider the problem of simulating a quantum channel $\mN\in\cptp(A\to B)$ using LOSE assisted with a noiseless classical communication resource. The classical communication cost of channel simulation under LOSE, denoted as $\cost^\eps(\mN)$, is the minimum number of one-bit noiseless classical channels, that Alice and Bob must use to approximate the action of $\mN$ up to an error $\eps\in(0,1)$.

The reverse quantum Shannon theorem states that for every $\eps\in(0,1)$
\be\label{rh}
\lim_{n\to\infty}\frac1n\cost^\eps\left(\mN^{\otimes n}\right)=I(A:B)_\mN\;.
\ee
Here, the mutual information of the channel $I(A:B)_\mN$ is obtained by maximizing the mutual information $I(A:B)_{\omega_\phi}$ over all pure states $\phi\in\md(A\tA)$ (with $\tA\cong A$), where $\omega_\phi^{AB}\eqdef\mN^{\tA\to B}(\phi^{A\tA})$.
The proof of this result is obtained by finding upper and lower bound for $\frac1n\cost^\eps\left(\mN^{\otimes n}\right)$ that converge to to the right-hand side of~\eqref{rh} in the asymptotic limit. Here we focus on finding an upper bound for $\cost^\eps\left(\mN^{\otimes n}\right)$, following ideas presented in~\cite{BCR2011} incorporating the universal upper bound we found in Theorem~\ref{uab}.

In the following theorem, for every $\alpha\in(0,1)$ and $\beta>1$  we denote by
\be
I_{\alpha,\beta}(A:B)_\mN\eqdef\max_{\phi}\Big\{H_{\alpha}(A)_{\omega_\phi}-\tH_{\beta}^\ua(A|B)_{\omega_\phi}\Big\}
\ee
where the maximum is over all pure states $\phi^{A\tA}$. Then, for every $n\in\mbb{N}$:
\bmyt\label{thmj}
\be
\frac1n\cost^\eps\left(\mN^{\otimes n}\right)\leq I_{\alpha,\beta}(A:B)_\mN+\mathcal{O}\left(\frac{\log (n)}n\right)
\ee
\emyt
Observe that Theorem~\ref{thmj} yields
\be\label{188}
\limsup_{n\to\infty}\frac1n\cost^{\eps}\left(\mN^{\otimes n}\right)\leq I_{\alpha,\beta}(A:B)_{\mN}\;.
\ee
Since this inequality holds for every $\alpha\in(0,1)$ and $\beta>1$ we get from the continuity of the function $(a,\beta)\mapsto I_{\alpha,\beta}(A:B)_{\mN}$ that the inequality in~\eqref{188} also holds for $\alpha=\beta=1$. That is,
\be
\limsup_{n\to\infty}\frac1n\cost^{\eps}\left(\mN^{\otimes n}\right)\leq I(A:B)_{\mN}\;.
\ee

\noindent\textbf{Conclusions:}
In this work, we introduced an equality-based convex split lemma by employing the collision relative entropy instead of the Umegaki relative entropy. This refinement strengthens the lemma by converting its conventional inequality into an exact relation. Beyond its immediate applications, this result provides a deeper understanding of the fundamental origin of the convex split lemma. The enhanced formulation improves its utility in key quantum information tasks, particularly quantum state splitting and quantum state redistribution, where tighter achievability bounds are crucial.

Additionally, we established a universal upper bound on the smoothed max mutual information, providing a dimension-independent, additive entropy bound. This result has direct implications for quantum communication theory, and we demonstrated its applicability in the reverse quantum Shannon theorem, where it helps quantify the resource cost of simulating quantum channels.

Together, these contributions refine fundamental tools in single-shot quantum information theory, leading to sharper insights into quantum resource trade-offs and optimal protocol design in both finite and asymptotic regimes.

\begin{acknowledgments}
\noindent\textbf{Acknowledgments:} The author sincerely thanks Kun Fang, Mark Wilde, and Andreas Winter for valuable discussions. Special thanks to Bartosz Regula for insightful discussions and for highlighting~\cite{RLD2025}, which significantly simplified the expression in~\eqref{0gag0} from an earlier version of this manuscript. This research was supported by the Israel Science Foundation under Grant No. 1192/24.
\end{acknowledgments}

\bibliographystyle{apsrev4-2}
\bibliography{QRT} 

\onecolumngrid

\appendix

\section*{\huge Appendix}
\setcounter{section}{0} 
\renewcommand{\thesection}{\Alph{section}} 

\section{Preliminaries}

\noindent{\it Notations:}
The letters $A$, $B$, $C$, and $R$ will be used to denote both quantum systems (or registers) and their corresponding Hilbert spaces, with $|A|$ representing the dimension of the Hilbert space associated with system $A$. Throughout, we consider only finite-dimensional Hilbert spaces. The letters $X$, $Y$, and $Z$ will be used to describe classical systems. To indicate a replica of a system, we use a tilde symbol above its label; for instance, $\tA$ denotes a replica of $A$, implying $|\tA| = |A|$. The set of all positive operators acting on system $A$ is denoted by $\pos(A)$, and the set of all density operators in $\pos(A)$ is denoted by $\md(A)$, with its subset of pure states represented by $\pure(A)$. The set of all effects in $\pos(A)$ (i.e., all operators $\Lambda\in\pos(A)$ with $\Lambda\leq I^A$) will be denoted as $\eff(A)$. The set of all quantum channels, i.e., completely positive and trace-preserving (CPTP) maps, from system $A$ to system $B$ is denoted by $\cptp(A \to B)$. $\prob(m)$ is the set of all probability vectors in $\mbb{R}^m$.
For every $\Lambda\in\pos(A)$ we use the notation $\spec(\Lambda)$ to denote the spectrum of $\Lambda$ (that is, the set of distinct eigenvalues of $\Lambda$).

\noindent{\it Properties:}
A useful property of $Q_\alpha$ (for all $\alpha\in[0,\infty]$) that we will employ in this paper is its behavior under direct sums. Specifically, consider two cq-states $\rho, \sigma \in \md(XA)$ of the form:
\be
\rho^{XA}\eqdef\sum_{x\in[k]}p_x|x\lr x|^X\otimes\rho_x^A
\ee
and
\be
\sigma^{XA}\eqdef\sum_{x\in[k]}p_x|x\lr x|^X\otimes\sigma_x^A\;.
\ee
where $\rho^X = \sigma^X$. For such states, $Q_\alpha$ satisfies the following \emph{direct sum property}:
\be\label{dsp}
Q_\alpha\left(\rho^{XA}\big\|\sigma^{XA}\right)=\sum_{x\in[k]}p_xQ_\alpha\left(\rho^A_x\big\|\sigma_x^A\right)\;.
\ee
This property will be useful to some of the results discussed in this paper.

\noindent{\it Collision Relative Entropy:}
In the classical case, where $\rho$ and $\sigma$ are replaced by probability vectors $\p$ and $\q$, the collision relative entropy is closely related to the $\chi^2$-distance. Specifically, it can be expressed as
\be
D_2(\p\|\q)=\log\left(1+\chi^2(\p\|\q)\right)\;,
\ee
where the $\chi^2$-distance is defined as
\be
\chi^2(\p\|\q)\eqdef\sum_{x\in[n]}\frac{(p_x-q_x)^2}{q_x}\;.
\ee
Using the inequality
\ba
\|\p-\q\|_1&=\sum_{x\in[n]}\frac{|p_x-q_x|}{\sqrt{q_x}}\sqrt{q_x}\\
\GG{Cauchy-Schwarz}&\leq\sqrt{\chi^2(\p\|\q)}
\ea
we can lower bound the collision relative entropy as
\be\label{661}
D_2(\p\|\q)\geq\log\left(1+\|\p-\q\|_1^2\right)\;.
\ee
This inequality highlights that when $D_2(\p\|\q)$ is very small, the total variation distance $\|\p - \q\|_1$ is also small, indicating that $\p$ is close to $\q$. Moreover, for the case that $\|\p - \q\|_1\leq 1$ this inequality is tight~\cite{Sason2016,RW2011} in the sense that for every $\eps\in(0,1)$ and $n\in\mbb{N}$,
\be
\min_{\substack{\p,\q\in\prob(n)\\\|\p-\q\|_1\geq\eps}}D_2(\p\|\q)=\log(1+\eps^2)\;.
\ee
For the case $1<\|\p - \q\|_1<2$, i.e., $\eps\in(1,2)$ one can get a better bound given by~\cite{Sason2016,RW2011} 
\be
\min_{\substack{\p,\q\in\prob(n)\\\|\p-\q\|_1\geq\eps}}D_2(\p\|\q)=-\log\left(1-\frac12\eps\right)\;.
\ee

The inequality~\eqref{661} can be generalized to the quantum case due to the existence of a quantum channel $\mE \in \cptp(A \to X)$, where $X$ is a classical system (see, e.g., Theorem~5.12 of~\cite{Gour2025}), such that
\be\label{662}
\|\rho-\sigma\|_1=\big\|\mE(\rho)-\mE(\sigma)\big\|_1\;.
\ee
Applying the data processing inequality (DPI) for $D_2$, we have
\ba\label{0a6}
D_2(\rho\|\sigma)&\geq D_2\left(\mE(\rho)\big\|\mE(\sigma)\right)\\
\GG{\eqref{661}}&\geq\log\left(1+\big\|\mE(\rho)-\mE(\sigma)\big\|_1^2\right)\\
\GG{\eqref{662}}&=\log\left(1+\|\rho-\sigma\|_1^2\right)\;.
\ea
The collision relative entropy can also be connected to the purified distance. Specifically, since $D_{1/2}$ is bounded above by $D_2$, it follows that for every $\rho, \sigma \in \md(A)$
\be
D_{2}(\rho\|\sigma)\geq-\log\left(1-P^2(\rho,\sigma)\right)\;.
\ee
This relation can be expressed as
\be
P^2(\rho,\sigma)\leq 1-\frac{1}{Q_2(\rho\|\sigma)}\;.
\ee

\section{Equality-Based Convex Split Lemma: Proofs and Additional Comments} 

\noindent\textbf{Lemma.}
{\it Let $n\in\mbb{N}$, $\rho\in\md(RA)$, $\sigma\in\md(A)$, $\omega\in\md(R)$, and $\tau^{RA^n}$ as in Eqs.~(\ref{rcyc},\ref{tau}).
Then,
\be\label{0gmain0}
Q_2\left(\tau^{RA^n}\big\|\omega^R\otimes\left(\sigma^{A}\right)^{\otimes n}\right)= \frac{n-1}nQ_{2}(\rho^{R}\|\omega^R)+\frac1nQ_{2}(\rho^{RA}\|\omega^R\otimes\sigma^A)\;.
\ee
}

\begin{proof}
For simplicity, we suppress superscripts wherever unnecessary. For $x \in [n]$, define
\be
\eta_x\eqdef\left(\omega\otimes\sigma^{\otimes n}\right)^{-\frac12}\tau_{x}\left(\omega\otimes\sigma^{\otimes n}\right)^{-\frac12}\;.
\ee
By definition,
\be\label{493}
Q_2\left(\tau\big\|\omega\otimes\sigma^{\otimes n}\right)=\frac1{n^2}\sum_{x,x'\in[n]}\tr\left[\tau_{x'}\eta_x\right]\;.
\ee
From the structure of $\tau_x$ in~\eqref{rcyc} and the property
\be
\left(\omega\otimes\sigma^{\otimes n}\right)^{-\frac12}=\omega^{-\frac12}\otimes\left(\sigma^{-\frac12}\right)^{\otimes n}\;,
\ee
we obtain 
\ba\label{etax}
\eta^{RA^n}_x=\left(\omega^R\otimes\sigma^{A_x}\right)^{-\frac12}\rho^{RA_x}\left(\omega^R\otimes\sigma^{A_x}\right)^{-\frac12}\otimes I^{A_1\cdots A_{x-1}A_{x+1}\cdots A_n}\;.
\ea
For $x \neq x'$, we have $\tau_{x'}^{RA_x} = \rho^R \otimes \sigma^{A_x}$. Thus, from the form of $\eta_x$ in~\eqref{etax}, it follows  that for $x\neq x'$,
 \be
 \tr\left[\tau_{x'} \eta_x\right] =Q_{2}(\rho^{R}\|\omega^R)\;.
 \ee 
 On the other hand, for $x = x'$, $\tau_{x}^{RA_x} = \rho^{RA_x}$, which gives 
 \be
 \tr\left[\tau_{x} \eta_x\right] = Q_{2}(\rho^{RA}\|\omega^R\otimes\sigma^A)\;.
 \ee 
 Substituting these expressions into~\eqref{493}, we conclude that
\ba
Q_2\left(\tau\big\|\rho\otimes\sigma^{\otimes n}\right)&=\frac1{n^2}\sum_{\substack{x\neq x'\\x,x'\in[n]}}Q_{2}(\rho^{R}\|\omega^R)+\frac1{n^2}\sum_{\substack{x= x'\\x,x'\in[n]}}Q_{2}(\rho^{RA}\|\omega^R\otimes\sigma^A)\\
&=\frac{n-1}nQ_{2}(\rho^{R}\|\omega^R)+\frac1nQ_{2}(\rho^{RA}\|\omega^R\otimes\sigma^A)\;.
\ea
This completes the proof.
\end{proof}

It is straightforward to verify that if we replace the definition of $\tau^{RA^n}$ in~\eqref{rcyc} with
\be
\tau^{RA^n}\eqdef\sum_{x\in[n]}p_x\tau^{RA^n}_x\;,
\ee
then the following holds:
\be\label{weighted}
Q_2\left(\tau^{RA^n}\big\|\omega^R\otimes\left(\sigma^{A}\right)^{\otimes n}\right)=(1-t)Q_{2}(\rho^{R}\|\omega^R) + tQ_{2}(\rho^{RA}\|\omega^R\otimes\sigma^A)\;,
\ee
where $t\eqdef\sum_{x\in[n]}p_x^2$.
Since for any probability distribution $\p\in\prob(n)$, we have
\be
t\eqdef\sum_{x\in[n]}p_x^2 \geq \frac{1}{n}\;,
\ee
it follows that the right-hand side in~\eqref{weighted} achieves its minimal value when $\p$ is the uniform distribution, $\p = \u = \left(\frac{1}{n}, \ldots, \frac{1}{n}\right)^T$.
Therefore, for the purposes of this discussion, we will always assume the use of the uniform distribution.\\

\subsubsection{Proof of Corollary~\ref{cor1}}

\noindent\textbf{Corollary.}
{\it Let $n\in\mbb{N}$, $\rho\in\md(RA)$, $\sigma\in\md(A)$, $\nu_n$ as in~\eqref{nu}, and $\tau^{RA^n}$ as in Eqs.~(\ref{rcyc},\ref{tau}). Then,
\be\label{split90} 
P^2\left(\tau^{RA^n}, \rho^R \otimes \left(\sigma^A\right)^{\otimes n}\right) \leq 1-\frac{1}{ \nu_n}\;. 
\ee 
}

\begin{proof}
Since ${D}_{1/2}$ is always less than or equal to the Umegaki relative entropy we get
\ba
-\log F^2\left(\tau, \rho \otimes \sigma^{\otimes n}\right)&\leq D\left(\tau\| \rho \otimes \sigma^{\otimes n}\right)\\
&=\min_{\omega\in\md(R)}D\left(\tau\| \omega \otimes \sigma^{\otimes n}\right)\\
&\leq \min_{\omega\in\md(R)}D_2\left(\tau\| \omega \otimes \sigma^{\otimes n}\right)\\
\GG{\eqref{gmain}}&\leq\log\left(\nu_n\right)\;.
\ea
Isolating $F^2$ and then substituting $F^2=1-P^2$ completes the proof.
\end{proof}

\noindent{\it Comparison with another variant of the convex split lemma:}
In~\cite{LY2024}, another variant of the convex split lemma has been proposed. In this variant, the right-hand side of~\eqref{gmain0} is replaced with another function that depends on the sandwiched relative entropy.

\noindent\textbf{Lemma.}\cite{LY2024}
{\it Let $n\in\mbb{N}$, $\rho\in\md(RA)$, $\sigma\in\md(A)$, and $\tau^{RA^n}$ as in Eqs.~(\ref{rcyc},\ref{tau}). 
Then,
\ba\label{ryr}
D\left(\tau^{RA^n}\big\|\rho^R\otimes\left(\sigma^{A}\right)^{\otimes n}\right)
\leq \frac{\ell^s}{sn^s} 2^{sD_{1+s}\left(\rho^{RA}\|\rho^R\otimes\sigma^A\right)}\;.
\ea
where $\ell\eqdef|\spec\left(\rho^R\otimes\sigma^A\right)|$.
}

The equality-based convex split lemma can also be used to derive an alternative bound, given by:
\be\label{imp}
D\left(\tau^{RA^n}\big\|\rho^R\otimes\left(\sigma^{A}\right)^{\otimes n}\right)\leq\log\left(1+\frac{\mu} n\right)\;,
\ee
where for simplicity we took $\omega^R=\rho^R$.
To compare this bound with~\eqref{ryr}, first observe that the inequality $\log(1+x)\leq x$ implies
\be\label{c}
\log\left(1+\frac{\mu} n\right)\leq\frac{\mu}{n}\leq\frac{1}{n}2^{D_{2}(\rho^{RA}\|\rho^R\otimes\sigma^R)}\;.
\ee
Now, observe that for sufficiently large $n$ we have
\be
\frac{1}{n}2^{D_{2}(\rho^{RA}\|\rho^R\otimes\sigma^R)}<\frac{\ell^s}{sn^s} 2^{sD_{1+s}\left(\rho^{RA}\|\rho^R\otimes\sigma^A\right)}
\ee
More precisely, for every $s\in[0,1]$ set 
\be
a_s\eqdef sD_{1+s}\left(\rho^{RA}\|\rho^R\otimes\sigma^A\right)\;.
\ee
Then, for every $n\in\mbb{N}$ that satisfies
\be\label{rh}
\log n>\frac{a_1-a_s-s\log (v)-\log\left(\frac1s\right)}{1-s}
\ee
the equality-based convex split lemma provides a tighter bound than the one given in~\eqref{ryr}. Note that the right-hand side of~\eqref{rh} becomes negative in the limits $s\to 0^+$ and $s\to 1^-$. Therefore, it achieves its maximum for some $0<s<1$.

Alternatively, if we apply the inequality $\log(1+x)\leq\frac{x^s}s$ instead of $\log(1+x)\leq x$, we get
\be\label{c}
\log\left(1+\frac{\mu} n\right)\leq\frac{\mu^s}{sn^s}\leq\frac{1}{sn^s}2^{sD_{2}(\rho^{RA}\|\rho^R\otimes\sigma^R)}\;.
\ee
Therefore, combining~\eqref{c} and~\eqref{ryr}, we deduce that if
\be
D_{1+s}\left(\rho^{RA}\|\rho^R\otimes\sigma^A\right)\geq D_{2}\left(\rho^{RA}\|\rho^R\otimes\sigma^A\right)-\log v
\ee
then the right-hand side of~\eqref{gmain} is smaller than the right-hand side of~\eqref{ryr}. This inequality holds in some cases, especially in the generic case where $v=\log|RA|$.

Nonetheless, in many other scenarios the bound in~\eqref{ryr} is more advantageous. This is particularly happen when considering multiple copies of $\rho^{RA}$, in which case the contribution of $v$ becomes negligible and the difference between $D_2$ and $D_{1+s}$ increases. Indeed,~\cite{LY2024} demonstrated this in specific applications, highlighting that their bound can perform very well in certain practical settings.

\section{Application: Quantum State Splitting}

In the two sections on applications, we apply our results to two key quantum information processing tasks: quantum state splitting and the reverse quantum Shannon theorem. By framing these tasks within the resource-theoretic framework~\cite{CG2019,Gour2025}, we provide a systematic and unified analysis that highlights their structural connections, and streamlines their characterization.

Quantum state splitting (QSS) is a source coding protocol that serves as the reverse or dual of quantum state merging (QSM). In QSM, Alice and Bob initially share a bipartite state $\rho^{AB}$, and the task is to transfer Alice's share, $A$, to Bob. In QSS, the process is reversed: consider an i.i.d. source producing composite pure states $\rho^{AA'}$, where both subsystems $A$ and $A'$ are held by Alice. The goal of QSS is to transfer the system $A'$ to Bob.
To formalize, let $B$ denote Bob's replica of $A'$. The objective of QSS is to simulate the action of the identity channel $\id^{AA' \to AB}$ on the purification $\rho^{RAA'}$ of the state $\rho^{AA'}$, where $R$ is the reference system (see Fig.~\ref{qss0}).

 In this section, we examine a setting where shared entanglement is freely available, but both classical and quantum communication are restricted. The permitted operations in this framework are therefore limited to local operations and shared entanglement (LOSE). However, LOSE alone is generally insufficient to implement the identity channel $\id^{AA' \to AB}$ on the purification $\rho^{RAA'}$. This raises the central question: what is the minimum number of qubits that Alice must transmit to Bob to approximate the action of $\id^{AA' \to AB}$ on $\rho^{RAA'}$ within an error tolerance of $\eps \in (0,1)$? Figure~\ref{lose0} illustrates the action of such an LOSE superchannel on a communication channel $\id_m$.\\
  
 \noindent{\it Resource Monotones:}
The set of free channels in $\cptp(A \to B)$ under LOSE is the set of replacement channels, defined as
\be
\mF_\omega^{A \to B}(\rho^A) = \tr[\rho^A] \omega^B \quad\quad \forall \; \rho \in \ml(A)\;,
\ee
where $\omega \in \md(B)$ is a fixed quantum state.  Therefore, for a resource channel $\mN \in \cptp(A \to B)$, the $\alpha$-sandwiched relative entropy of the resource (with free operations $\mf = {\rm LOSE}$) is defined as \ba\label{switch}
 D_\alpha(\mN\|\mf)&\eqdef\min_{\omega\in\md(B)}D_\alpha(\mN\|\mF_\omega)\\
 &\eqdef\min_{\omega\in\md(B)}\max_{\psi^{A\tA}}D_\alpha\left(\mN^{\tA\to B}\left(\psi^{A\tA}\right)\Big\|\mF_\omega^{\tA\to B}\left(\psi^{A\tA}\right)\right)\;.
 \ea
 where the maximization is over all pure states $\psi \in \pure(A \tA)$.
 
 As shown in~\cite{GW2019}, the order of the min and max in~\eqref{switch} can be interchanged. Additionally, since $\mF_\omega$ is a replacement channel,
\be
\mF_\omega^{\tA \to B}(\psi^{A \tA}) = \psi^A \otimes \omega^B\;,
\ee
where $\psi^A$ is the reduced density matrix of $\psi^{A \tA}$. Substituting this into \eqref{switch}, we obtain
 \ba\label{exp}
 D_\alpha(\mN\|\mf)
 &=\max_{\psi^{A\tA}}\min_{\omega\in\md(B)}D_\alpha\left(\mN^{\tA\to B}\left(\psi^{A\tA}\right)\Big\|\psi^{A}\otimes\omega^B \right)\\
 &=\max_{\psi^{A\tA}}I_\alpha(A:B)_{\mN^{\tA\to B}\left(\psi^{A\tA}\right)}
 \ea
 where $I_\alpha$ is the $\alpha$-R\'enyi mutual information defined in~\eqref{alpha}. Using the expression in~\eqref{exp}, we extend the definition of $\alpha$-R\'enyi mutual information to a channel $\mN \in \cptp(A \to B)$ as:
\be
I_\alpha(A:B)_\mN\eqdef\max_{\psi\in\pure(A\tA)}I_\alpha(A:B)_{\mN^{\tA\to B}\left(\psi^{A\tA}\right)}\;.
\ee
Thus, the $\alpha$-sandwiched relative entropy of a resource is equivalent to the channel $\alpha$-R\'enyi mutual information:
\be
D_\alpha(\mN\|\mf)=I_\alpha(A:B)_\mN\;.
\ee
In this work, we will see that the channel mutual information, particularly for the case $\alpha = 2$ (known as the collision mutual information), plays a crucial operational role in single-shot quantum state splitting.\\

\noindent{\it The Conversion Distance:}
The conversion distance associate with such a task is defined as:
\be
T_\rho\left(\id_m\xrightarrow{{\rm LOSE}} \id^{AA'\to AB}\right)\eqdef\min_{\Theta\in\lose} \frac12\left\|\mN^{AA'\to A B}_\Theta\left(\rho^{RAA'}\right)-\rho^{RAB}\right\|_1\;, 
\ee
where $\mN^{AA'\to A B}_\Theta\eqdef\Theta\left[\id_m\right]$, and the minimum is over all LOSE superchannels $\Theta$. The LOSE superchannel $\Theta$ is defined such that (see Fig.~\ref{lose0})
\be\label{thet}
\mN^{AA'\to A B}_\Theta\left(\rho^{RAA'}\right)=\mF^{B_0B_1\to B}\circ\id_m^{A_0\to B_0}\circ\mE^{AA'A_1\to AA_0}\left(\rho^{RAA'}\otimes\varphi^{A_1B_1}\right)
\ee
Thus, optimization over all $\Theta$ amounts to optimization over all finite dimensional systems $A_1$ and $B_1$ (w.l.o.g.\ we can assume that $k\eqdef|A_1|=|B_1|$), all $\mE$ and $\mF$, and all entangled states $\varphi$.\\

\noindent{\it The Optimal Communication Cost:}
Given $\eps \in (0,1)$ and $m \in \mbb{N}$, an LOSE superchannel $\Theta$, as shown in Fig.~\ref{lose0}, is called an $(\eps,m)$-QSS protocol if $\mN_\Theta$ (as defined in~\eqref{thet}) has the property that $\mN^{AA'\to A B}_\Theta\left(\rho^{RAA'}\right)$ is $\eps$-close to $\rho^{RAB}$. Therefore, the communication cost for an $\eps$-error QSS (under LOSE) is the minimum communication cost over all such $(\eps,m)$-QSS protocols. We summarize it in the following definition.

\bmyd
Let $\eps\in(0,1)$ and $\rho\in\md(AA')$. 
The communication cost for an $\eps$-error QSS under LOSE is defined as
\be\label{epsqss}
\cost^\eps\left(\rho^{AA'}\right)\eqdef\inf\log (m)\;,
\ee
where the infimum is over all $m\in\mbb{N}$ that satisfy
\be
T_\rho\left(\id_m\xrightarrow{{\rm LOSE}} \id^{AA'\to AB}\right)\leq\eps\;.
\ee
\emyd

Our goal in this section is to calculate the minimal quantum communication cost necessary to implement a QSS protocol. In~\cite{BCR2011} it was shown that for $\eps\in(0,1)$ and $\rho\in\pure(RAA')$ the communication cost for an $\eps$-error QSS under LOSE is lower bounded by
\be\label{lb}
\cost^\eps\left(\rho^{AA'}\right)\geq\frac12I_{\max}^\eps(A':R)_\rho\;.
\ee
In~\cite{BCR2011} and subsequently in~\cite{ADJ2017}, an upper bound was derived, demonstrating that
\be
\cost^{2\eps}\left(\rho^{AA'}\right)\leq\frac12I_{\max}^\eps(R:A')_\rho+\log\left(\frac2\eps\right)\;
\ee
Here we show that one can replace the max mutual information in the upper bound above with the smaller collision mutual information.
 To do this, we construct the same $\eps$-error QSS protocol introduced in~\cite{ADJ2017} but instead utilizing the equality-based convex split lemma.\\

\noindent{\it Improved Achievability Bound:}
We consider a state $\rho\in\pure(RAA')$, in which $R$ is a reference system, and $A$ and $A'$ are systems held by Alice. Without loss of generality, we will assume that $|A'|=|A|$ since otherwise, if necessary, we will embed the pure state $\rho^{RAA'}$ in a larger dimensional space in which the corresponding systems $A$ and $A'$ have the same dimensions. We are now ready to prove of Theorem~\ref{thmeqss}.

\noindent\textbf{Theorem.}
{\it Let $0<\delta<\eps<1$  and $\rho\in\pure(RAA')$. The communication cost for an $\eps$-error QSS under LOSE is upper bounded by
\be\label{ub00}
\cost^{\eps}\left(\rho^{AA'}\right)\leq\frac12I_{2}^{\eps-\delta}(R:A')_\rho+\log\left(\frac1\delta\right)\;.
\ee}

\begin{proof}
The proof closely follows the methodology outlined in~\cite{ADJ2017}, with the key distinction being the application of the quality-based quantum split lemma.
Let $B$ be a replica of $A'$ on Bob's side. Our goal is to construct an LOSE protocol simulating the action of the channel $\id^{A'\to B}$ on the state $\rho^{RAA'}$. Denoting by $\rho^{RAB}\eqdef\id^{A'\to B}\left(\rho^{RAA'}\right)$ we define $\sigma\in\md(B)$ to be an optimal state that satisfies
\be\label{sati00}
I_{2}(R:B)_\rho=D_{2}\left(\rho^{RB}\big\|\rho^{R}\otimes\sigma^{B}\right)\;.
\ee
The key idea of the proof is to fix an integer $n\in\mbb{N}$ (to be determined shortly) and invoke the convex split lemma for the state $\rho^{R}\otimes(\sigma^B)^{\otimes n}$. Specifically, using the notation $B^n=(B_1,\ldots,B_n)$ for $n$ copies of $B$, the convex split lemma (Lemma~\ref{split}),
particularly as given in~\eqref{split9}, states that the purified distance between $\rho^{R}\otimes(\sigma^B)^{\otimes n}$ and the state
\ba
\tau^{RB^n}\eqdef\frac1n\sum_{x\in[n]}
\tau_x^{RB_n}\;,
\ea
with
\be
\tau_x^{RB_n}\eqdef\rho^{RB_x}\otimes\sigma^{B_1}\otimes\cdots\otimes\sigma^{B_{x-1}}
\otimes\sigma^{B_{x+1}}\otimes\cdots\otimes\sigma^{B_n}\;,
\ee
is no greater than 
$\delta_n\eqdef\sqrt{1-\frac{1}{\nu_n}}$, where $\nu_n$ is defined as in~\eqref{nu} but with system $A$ replaced with $B$. 
Moreover, from Uhlmann theorem this purified distance can be expressed in terms of pure states. Explicitely, let $\phi^{AB}$ be a purification of $\sigma^{B}$ (we assume $|A|=|B|$), and consider the following purification of $\tau^{RB^n}$ given by
\be
\big|\tau^{R(LA^n)(B^n)}\big\ra\eqdef\frac1{\sqrt{n}}\sum_{x\in[n]}|x\ra^L\big|\varphi^{RA^n B^n }_x\big\ra\;,
\ee
where for every $x\in[n]$
\be\label{defl00}
\varphi_x^{RA^nC^n}\eqdef\rho^{RA_x B_x}\otimes\phi^{A_1B_1}\otimes\cdots\otimes\phi^{A_{x-1}B_{x-1}}
\otimes\phi^{A_{x+1}B_{x+1}}\otimes\cdots\otimes\phi^{A_nB_n}\;.
\ee
Since $\rho^{RAA'}\otimes\left(\phi^{AB}\right)^{\otimes n}$ is a purification of $\rho^{R}\otimes\left(\sigma^B\right)^{\otimes n}$ we get from Uhlmann theorem that there exists an isometry channel $\mV\in\cptp(AA'A^n\to LA^n)$ such that
\be\label{pmu200}
P\left(\tau^{R(LA^n)(B^n)},\mV\left(\rho^{RAA'}\otimes\left(\phi^{AB}\right)^{\otimes n}\right)\right)=P\left(\tau^{RB^n},\rho^{R}\otimes\left(\sigma^B\right)^{\otimes n}\right)\leq\delta_n\;.
\ee
With this in mind, consider the following LOSE protocol comprising of the four steps:
\ben
\item Alice and Bob borrow $n$ copies of the (entangled) state $\phi^{AB}$. The initial state is therefore $\rho^{RAA'}\otimes\phi^{\otimes n}$.
\item Alice apply the isometry channel $\mV\in\cptp(AA'A^n\to LA^n)$ to her systems, resulting in the state  $\mV(\rho^{RAA'}\otimes\phi^{\otimes n})$.
\item Alice applies a basis measurement on system $L$ in the basis $\{|x\ra^L\}_{x\in[n]}$, and communicate the outcome $x$ of the measurement to Bob. 
\item Alice swap the system (register)  $A_x$ with $A_1\equiv A$, and Bob swap the system $B_x$ with  $B_1\equiv B$
\een

Since on pure states, the trace distance equals the purified distance,  we get from~\eqref{pmu200} that the state $\mV(\rho^{RAA'} \otimes \phi^{\otimes n})$ at the second step is $\delta_n$-close to $\tau^{R(LA^n)B^n}$. By the data-processing inequality  for the trace distance (or purified distance), it follows that the final state obtained after completing the remaining steps of the protocol will also be $\delta_n$-close to the state resulting from applying those same final steps to $\tau^{R(LA^n)B^n}$.

Now, by applying Alice measurement of the third step to the state $\tau^{R(LA^n)B^n}$, after outcome $x$ occurred the resulting state is $\varphi_x^{RA^nB^n}$. From the form of $\varphi_x$ in~\eqref{defl00} we see that if Alice swap the system  $A_x$ with $A_1\equiv A$, and Bob swap the system $B_x$ with  $B_1\equiv B$, as outlined in step 4 of the protocol, the state $\varphi_x^{RA^nB^n}$ transforms to $\rho^{RAB}\otimes(\phi^{AB})^{\otimes (n-1)}$. In other words, with local operations assisted with $\log(n)$ bits of classical communication (or equivalently, due to superdense coding,  $\frac12\log n$ bits of quantum communication), Alice and Bob can transform the state $\tau^{R(LA^n)B^n}$ to the state $\rho^{RAB}\otimes(\phi^{AB})^{\otimes (n-1)}$. We therefore conclude that the final state of the protocol is $\delta_n$-close to the (desired) state $\rho^{RAB}\otimes(\phi^{AB})^{\otimes (n-1)}$.

Finally, we choose $n$ such that the accuracy of the protocol is given by $\delta_n\leq\delta$. Set $\mu\eqdef Q_{2}(\rho^{RB}\|\rho^R\otimes\sigma^B)-1$ and observe that from~\eqref{nu2} we get that $\nu_n\leq1+\frac\mu{n}$, so that $\delta_n\leq\sqrt{\frac\mu{\mu+n}}$.
We therefore choose $n$ to be the smallest number satisfying $\sqrt{\frac\mu{\mu+n}}\leq\delta$. That is,
\be
n\eqdef\left\lceil\mu\left(\frac{1}{\delta^2}-1\right)\right\rceil\leq\frac{\mu+1}{\delta^2}\;.
\ee
Thus, the quantum communication cost of this protocol is bounded by
\ba\label{term}
\frac12\log n&\leq \frac12\log(\mu+1)+\log\left(\frac1\delta\right)\\
&=\frac12I_{2}(R:A')_\rho+\log\left(\frac1\delta\right)\;.
\ea

In the final stage, we replace the mutual information of order 2 (as it appear in~\eqref{term}) with its smooth counterpart. This adjustment allows us to lower the quantum communication cost by modifying our strategy: rather than applying the protocol directly to the state $\rho^{RAA'}$, we assume an alternative state ${\rho'}^{RAA'}$ that is $(\eps-\delta)$-close to $\rho^{RAA'}$, and then apply the protocol to ${\rho'}^{RAA'}$. Due to the data processing inequality, the additional error caused by this modification is bounded above by $\eps-\delta$.
Utilizing the triangle inequality, the overall accuracy of the protocol is given by $(\eps-\delta) + \delta=\eps$. By optimizing over all states ${\rho'}^{RAA'}$ that are $(\eps-\delta)$-close to $\rho^{RAA'}$, we can achieve a further reduction in the quantum communication cost, resulting in~\eqref{ub00}.
This completes the proof.
\end{proof}

\section{Universal Upper Bound on Smoothed Max Information}

A cornerstone of single-shot quantum information theory is the use of smoothed entropic functions and divergences to express optimal single-shot rates for quantum tasks. These rates are then bounded from above and below using additive functions, enabling a seamless transition to asymptotic rates.

Consider the hypothesis testing divergence, defined in~\eqref{htd} for any $\eps \in (0,1)$ and for two density matrices $\rho, \sigma \in \md(A)$. This divergence plays a pivotal role, with operational interpretations in areas such as quantum hypothesis testing, thermodynamics, and resource theories~\cite{CG2019, Gour2025}.

A lower bound for the hypothesis testing divergence was derived in~\cite{FGW2022} (and earlier in a weaker form in~\cite{QWW2018}) for all $\alpha \in (0,1)$: 
\be\label{htda}
D_{\min}^\eps(\rho\|\sigma)\geq D_\alpha(\rho\|\sigma)+\frac\alpha{1-\alpha}\left(\frac{h(\alpha)}{\alpha}-\log\left(\frac1\eps\right)\right)\;,
\ee
where $h(\alpha) \eqdef -\alpha\log\alpha - (1-\alpha)\log(1-\alpha)$.

An upper bound, valid for all $\beta>1$, was established in~\cite{CMW2016}:
\be\label{betab}
D_{\min}^\eps(\rho\|\sigma)\leq D_\beta(\rho\|\sigma)+\frac\beta{\beta-1}\log\left(\frac1{1-\eps}\right)\;.
\ee
Notably, both bounds are independent of the dimensions of the underlying Hilbert space, and exhibit additivity under tensor products, up to correction terms depending only on $\eps$, $\alpha$, and $\beta$.

These universal bounds, such as Eqs.~\eqref{htda} and~\eqref{betab}, can also be extended to other entropic quantities. For instance, consider the smoothed max-relative entropy: 
\be
D_{\max}^\eps(\rho\|\sigma)\eqdef\min_{\trho\in\mb^\eps(\rho)}D_{\max}(\trho\|\sigma)\;.
\ee 
In~\cite{RLD2025} it was shown that
\be\label{direct2}
D_{\max}^{\eps}(\rho\|\sigma)\leq D_\beta(\rho\|\sigma)+\frac1{\beta-1}\log\left(\frac1{\eps^2}\right)\;.
\ee
As a direct consequence of this upper bound, we obtain a universal lower bound on the optimized conditional min-entropy: 
\ba\label{direct}
H_{\min}^{\eps}(A|B)_\rho&\eqdef-\min_{\sigma\in\md(B)}D_{\max}^\eps\left(\rho^{AB}\big\|I^A\otimes\sigma^B\right)\\
\GG{\eqref{direct2}}&\geq \tH_{\beta}^{\ua}(A|B)_\rho-\frac1{\beta-1}\log\left(\frac1{\eps^2}\right)\;.
\ea

\subsubsection*{The Smoothed Max Information}

We now turn our attention to the smoothed max-information, a pivotal quantity in numerous QIP tasks. Unlike the case of the max-relative entropy, the techniques employed earlier cannot be directly applied to derive universal bounds for the smoothed max-information. Nevertheless, given its critical role in various operational scenarios, establishing universal bounds for this quantity remains a topic of significant interest and practical relevance.

Every relative entropy (i.e., additive quantum divergence) can be employed to define other entropic functions. In the context of the applications considered in this paper, we frequently encounter the mutual information. For every $\alpha\in[0,\infty]$, we define the $\alpha$-mutual information of a bipartite state as follows:
\be\label{alpha} 
I_\alpha(A:B)_\rho\eqdef\min_{\sigma\in\md(B)}D_\alpha\left(\rho^{AB}\big\|\rho^A\otimes\sigma^B\right)\;. 
\ee 
While alternative definitions of mutual information have been proposed in the literature (see, for example,~\cite{CBR2014} and references therein), often with various operational interpretations, we adopt the above definition as it is the most relevant to the applications considered in this paper.

In particular, we focus on three notable special cases of the mutual information:
\ben \item \textbf{The case $\alpha=1$:} Here, $I_1(A:B)_\rho$, often referred to simply as $I(A:B)_\rho$, is expressed as
\be 
I(A:B)_\rho=D\left(\rho^{AB}\big\|\rho^A\otimes\rho^B\right)\;. 
\ee 
\item \textbf{The case $\alpha=2$:} This case is given by
\be\label{i2}
I_2(A:B)_\rho\eqdef\min_{\sigma\in\md(B)}\log Q_2\left(\rho^{AB}\big\|\rho^A\otimes\sigma^B\right)\;. 
\ee 
\item \textbf{The case $\alpha=\infty$:} This is expressed in terms of the max-relative entropy as
\be 
I_{\max}(A:B)_\rho\eqdef\min_{\sigma\in\md(B)}D_{\max}\left(\rho^{AB}\big\|\rho^A\otimes\sigma^B\right)\;. 
\ee 
\een

We also consider the smoothed version of the max mutual information, defined for all $\eps\in(0,1)$ and $\rho\in\md(AB)$ as:
\be 
I_{\max}^\eps(A:B)_\rho\eqdef\min_{\rho'\in\mb^\eps(\rho)}I_{\max}(A:B)_{\rho'}\;. 
\ee 
For the smoothed max-mutual information, it is not possible to apply~\eqref{direct2}, as seen from its definition: 
\be 
I_{\max}^\eps(A:B)_\rho=\min_{\substack{\trho\in\mb^\eps(\rho^{AB})\\\sigma\in\md(B)}}D_{\max}\left(\trho^{AB}\big\|\trho^A\otimes\sigma^B\right)\;. 
\ee 
In other words, $I_{\max}^\eps$ does not directly depend on $D_{\max}^\eps$, and therefore, we cannot use~\eqref{direct2} to obtain a universal upper bound for the smoothed max information. Consequently, an alternative approach is required to establish a universal upper bound for $I_{\max}^\eps$.

In the following theorem, we derive such an upper bound. For every $\alpha\in(0,1)$ and $\beta>1$, let $c\eqdef 2-\sqrt{3}\approx 0.27$, and define: 
\be\label{alep5}
f_{\alpha,\beta}(\eps)\eqdef \left(\frac2{\beta-1}+\frac1{1-\alpha}\right)\log\left(\frac1{c\eps^2}\right)\;.
\ee

\subsubsection{Proof of Theorem~\ref{uab}}\label{Auniversal}

To establish the theorem, we begin by proving several key properties, starting with an important property of smoothing in terms of trace distance.

 \begin{lemma}\label{still}
{\it Let $\rho,\sigma\in\md(A)$, $d\eqdef|A|$, and denote by $\p,\q\in\prob^\da(d)$ the probability vectors whose components are the eigenvalues of $\rho$ and $\sigma$, respectively, arranged in non-increasing order. Then,
\be\label{eqin}
\min_{U\in\mfu(A)}\frac12\left\|\rho-U\sigma U^*\right\|_1=\frac12\|\p-\q\|_1\;.
\ee}
\end{lemma}

\begin{proof}
For every operator $\Gamma\in\herm(A)$, we use the notation $\{\lambda_k^\da(\Gamma)\}_{k\in[d]}$ to denotes the eigenvalues of $\Gamma$ arranged in non-increasing order. Observe that the $k$th largest eigenvalue of an operator $\Gamma\in\herm(A)$, that is, $\lambda_k^\da(\Gamma)$, can be expressed as
\be\label{lkth}
\lambda^\da_k\left(\Gamma\right)=\max_{\substack{B\subseteq A\\ \dim(B)=k}}\min_{\psi\in\pure(B)}\la\psi|\Gamma|\psi\ra\;,
\ee
where the maximum is over all subspaces $B$ of $A$ of dimension $\dim(B)=k$. Combining this expression for $\lambda^\da_k$ with the fact that for every positive operator $\Lambda\in\pos(A)$ and every $\psi\in\pure(B)$ (or even for all $\psi\in\pure(A)$) we have that $\la\psi|\Gamma+\Lambda|\psi\ra$ is no smaller than $\la\psi|\Gamma|\psi\ra$, we obtain that
\be\label{tak0}
\lambda^\da_k\left(\Gamma+\Lambda\right)\geq\lambda^\da_k\left(\Gamma\right)\quad\quad\forall\;\Gamma,\Lambda\in\pos(A)\;.
\ee 
Finally, taking in~\eqref{tak0}, $\Gamma=\rho$ and $\Lambda=(\rho-\sigma)_-$, we obtain
\be
\lambda^\da_k\big(\rho+(\rho-\sigma)_-\big)\geq\lambda^\da_k(\rho)=p_k\;.
\ee
Similarly, taking in~\eqref{tak0}, $\Gamma=\sigma$ and $\Lambda=(\rho-\sigma)_+$, we obtain
\be
\lambda^\da_k\big(\sigma+(\rho-\sigma)_+\big)\geq\lambda^\da_k(\sigma)=q_k\;.
\ee
But observe that $\rho+(\rho-\sigma)_-=\sigma+(\rho-\sigma)_+$, so that from the two equations above we get
\be\label{x0x}
\lambda^\da_k\big(\rho+(\rho-\sigma)_-\big)\geq\max\{p_k,q_k\}\;.
\ee
With this at hand, we conclude
\ba
\frac12\|\rho-\sigma\|_1=\tr(\rho-\sigma)_+
&=\tr\left[\rho+(\rho-\sigma)_-\right]-1\\
\GG{\eqref{x0x}}&\geq\sum_{k\in[d]}\max\{p_k,q_k\}-1\\
\Gg{2\max\{a,b\}=a+b+|a-b|}&\geq\frac12\sum_{k\in[d]}\big(p_k+q_k+|p_k-q_k|\big)-1\\
&=\frac12\|\p-\q\|_1\;.
\ea
This completes the proof.
\end{proof}

Consider a function $f: \md(A) \to \mbb{R}$, such as entropy, that is invariant under unitary channels; i.e., $f(U\sigma U^*)=f(\sigma)$ for all $\sigma\in\md(A)$ and $U\in\mfu(A)$. In this case, there exists a function $g: \prob(d) \to \mbb{R}$, where $d \eqdef |A|$, such that for every $\rho \in \md(A)$, we have $f(\rho) = g(\p)$, where $\p$ is the probability vector composed of the eigenvalues of $\rho$. 
For functions of this type, it is convenient to identify a metric where smoothing $f$ with respect to this metric is equivalent to smoothing $g$ with respect to the \emph{same} metric when restricted to probability vectors (i.e., diagonal density matrices). The Lemma~\ref{still} demonstrates that the trace distance has precisely this property. It can be shown (in fact even simpler to prove) that the also the purified distance has this property~\cite{BCR2011}.

In the following corollary, for every $\eps\in(0,1)$ we denote the smoothed version of $f$ by
\be\label{feps}
f^\eps(\rho)\eqdef\min_{\sigma\in\mb^\eps(\rho)}f(\sigma)\;,
\ee
where
\be\label{mbeps}
\mb^\eps\left(\rho\right)\eqdef\left\{\sigma\in\md(A)\;:\;\frac12\left\|\rho-\sigma\right\|_1\leq\eps\right\}\;.
\ee
We will also identify $\p$ as a diagonal density matrix, allowing us to write $f(\rho) = f(\p)$. This eliminates the need to introduce the function $g$, thus simplifying the exposition. 

\begin{corollary}\label{coru}
{\it Let $\eps\in(0,1)$, $\rho\in\md(A)$, $d\eqdef|A|$, $\p\in\prob^\da(d)$ be the probability vector composed of the eigenvalues of $\rho$, $f:\md(A)\to\mbb{R}$ be a function invariant under unitary channels, and $f^\eps(\rho)$ as in~\eqref{feps}. Then, 
\be\label{192}
f^\eps(\rho)=\min_{\q\in\prob^\da(d)}\left\{f(\q)\;:\;\frac12\|\p-\q\|_1\leq\eps\right\}\;.
\ee }
\end{corollary}

It is important to note that we assume that the components of both $\p$ and $\q$ are arranged in non-increasing order.

\begin{proof}
Denote by $R$ the left-hand and right-hand sides of~\eqref{192}, respectively. We need to show that $f^\eps(\rho)=R$. Since we identify probability vectors with diagonal density matrices we can view $\prob^\da(d)$ as a subset of $\md(A)$. Thus, clearly $f^\eps(\rho)\leq R$. To prove the opposite inequality, observe that from the inequality $\|\rho-\sigma\|_1\geq \|\p-\q\|_1$, that follows from Lemma~\eqref{still}, we get
\be
\min_{\sigma\in\md(A)}\left\{f(\sigma)\;:\;\frac12\|\rho-\sigma\|_1\leq\eps\right\}\geq \min_{\sigma\in\md(A)}\left\{f(\sigma)\;:\;\frac12\|\p-\q\|_1\leq\eps\right\}=R\;.
\ee
Hence, $f^\eps(\rho)\geq R$ and since we also have $f^\eps(\rho)\leq R$ we conclude that $f^\eps(\rho)= R$.
This completes the proof.
\end{proof}

The following lemma establishes a connection between the max mutual information and the optimized conditional min-entropy. For every $\rho \in \md(AB)$, the optimized conditional min-entropy is defined as 
\be\label{optimizer}
H_{\min}(A|B)_\rho=-\min_{\sigma\in\md(B)}D_{\max}\left(\rho^{AB}\big\|I^A\otimes\sigma^B\right)\;.
\ee
In this lemma, we use the notation that, for any $\Lambda \in \pos(A)$, $\lambda_{\min}(\Lambda)$ represents the smallest \emph{non-zero} eigenvalue of $\Lambda$. The lemma, originally proved in~\cite{BCR2011}, is restated here for completeness, accompanied by its short proof.

\begin{lemma}\label{epszero}\cite{BCR2011}
{\it Let $\rho\in\md(AB)$. Then,
\be
I_{\max}(A:B)_\rho\leq -\log\lambda_{\min}\left(\rho^A\right)-H_{\min}(A|B)_\rho
\ee}
\end{lemma}
\begin{proof}
Let $\sigma\in\md(B)$ be an optimizer of~\eqref{optimizer}. 
Then,
\be\label{upper}
I_{\max}(A:B)_\rho\leq D_{\max}\left(\rho^{AB}\big\|\rho^A\otimes \sigma^B\right)\;.
\ee
Next, we use the relation $\rho^A\geq\lambda_{\min}(\rho^A)\Pi^A_\rho$, where $\lambda_{\min}(\rho^A)$ denotes the smallest \emph{non-zero} eigenvalue of $\rho^A$ and $\Pi_\rho^A$ is the projection to the support of $\rho^A$. Combining this with~\eqref{upper} and the fact that $D_{\max}$ is L\"owner-decreasing in its second argument (Lemma~\ref{lowner}) we obtain
\ba\label{rhso}
I_{\max}(A:B)_\rho&\leq D_{\max}\left(\rho^{AB}\big\|\lambda_{\min}(\rho^A)\Pi^A_\rho\otimes \sigma^B\right)\\
&=-\log\lambda_{\min}(\rho^A)+D_{\max}\left(\rho^{AB}\big\|\Pi^A_\rho\otimes \sigma^B\right)\;.
\ea
Finally, the proof is concluded by recognizing that $\Pi_\rho^A$ as it appears on the right-hand side of~\eqref{rhso} can be replaced with $I^A$. This completes the proof.
\end{proof}

We are now ready to prove Theorem~\ref{uab}.

\noindent\textbf{Theorem.}
{\it Let $\eps,\alpha\in(0,1)$, $\beta>1$, and $\rho\in\md(AB)$. Then: \ba\label{final0} I_{\max}^\eps(A:B)_\rho&\leq H_\alpha(A)_{\rho}-\tH_{\beta}^{\ua}(A|B)_\rho+f_{\alpha,\beta}(\eps)\;, 
\ea 
where $f_\eps(\alpha,\beta)$ is defined in~\eqref{alep5}. 
}

\begin{proof}
From Lemma~\ref{epszero} it follows that
\be
I_{\max}^\eps(A:B)_\rho\leq \min_{\trho\in\mb^\eps(\rho^{AB})}\Big\{-\log\lambda_{\min}\left(\trho^A\right)-H_{\min}^\ua(A|B)_{\trho}\Big\}
\ee
Since the minimization above is applied to both terms, we need a strategy to upper bound the right-hand side. For this reason, we split the minimization into the following two parts: 
\ben
\item Minimization over all states $\omega^{AB}$ that are $\eps_0$-close to $\rho^{AB}$ for some $\eps_0\in(0,1)$.
\item Minimization is over all effects $\Lambda\in\eff(A)$ such that
\be\label{ol}
\omega_\Lambda^{AB}\eqdef\frac{\Lambda^A\omega^{AB}\Lambda^A}{\tr\left[\left(\Lambda^A\right)^2\omega^A\right]}
\ee
is $\eps_1$-close to $\omega^{AB}$, for another $\eps_1\in(0,1)$. 
\een
Note that from the triangle inequality we get that $\omega_\Lambda^{AB}$ is $(\eps_0+\eps_1)$-close to $\rho^{AB}$. Hence, we will choose $\eps_0,\eps_1\in(0,1)$ that satisfies $\eps_0+\eps_1\leq\eps$. With this in mind,
\be\label{splitt}
I_{\max}^\eps(A:B)_\rho\leq \min_{\omega\in\mb^{\eps_0}(\rho^{AB})}\min_{\substack{\Lambda\in\eff(A)\\\omega_\Lambda\approx_{\eps_{1}}\omega}}\Big\{-\log\lambda_{\min}\left(\omega^A_\Lambda\right)-H_{\min}(A|B)_{\omega_\Lambda}\Big\}
\ee
The normalization factor $\tr[(\Lambda^A)^2\omega^A]$ in~\eqref{ol} yields the same contribution in $\log\lambda_{\min}\left(\omega^A_\Lambda\right)$ and $H_{\min}(A|B)_{\omega_\Lambda}$ so it
cancel out in~\eqref{splitt}. Moreover, we choose
$\omega\in\md(AB)$ to be the optimal density matrix that satisfies $H_{\min}(A|B)_\omega=H_{\min}^{\eps_0}(A|B)_\rho$. Combining all this we get that
\be\label{lol}
I_{\max}^\eps(A:B)_\rho\leq \min_{\substack{\Lambda\in\eff(A)\\\omega_\Lambda\approx_{\eps_{1}}\omega}}\Big\{-\log\lambda_{\min}\left(\Lambda^A\omega^A\Lambda^A\right)-H_{\min}^\ua(A|B)_{\Lambda^A\omega^{AB}\Lambda^A}\Big\}\;.
\ee

Next, we simplify the conditional min-entropy that appear on the right-hand side. Let $\sigma\in\md(B)$ be such that $-H_{\min}(A|B)_{\omega}=D_{\max}\left(\omega^{AB}\big\|I^A\otimes\sigma^B\right)$. By definition,
\be\label{mh}
-H_{\min}^\ua(A|B)_{\Lambda^A\omega\Lambda^A}\leq D_{\max}\left(\Lambda^A\omega^{AB}\Lambda^A\big\|I^A\otimes\sigma^B\right)
\ee
Now, observe that since $0\leq \Lambda^A\leq I^A$ we get that 
\ba
I^A\otimes\sigma^B&\geq (\Lambda^A)^2\otimes\sigma^B\\
&=\Lambda^A\left(I^A\otimes\sigma^B\right)\Lambda^A\;.
\ea
Combining this with fact that $D_{\max}$ is L\"owner-decreasing in its second argument we get from~\eqref{mh} that
\ba
-H_{\min}^\ua(A|B)_{\Lambda^A\omega\Lambda^A}&\leq D_{\max}\left(\Lambda^A\omega^{AB}\Lambda^A\big\|\Lambda^A\left(I^A\otimes\sigma^B\right)\Lambda^A\right)\\
\GG{DPI}&\leq D_{\max}\left(\omega^{AB}\big\|I^A\otimes\sigma^B\right)\\
\GG{\text{By definition of }\sigma^{\it B}}&=-H_{\min}^\ua(A|B)_\omega\;,
\ea
where we used the generalized DPI property of $D_{\max}$.
Substituting this inequality into~\eqref{lol}  we obtain
\be\label{from0}
I_{\max}^\eps(A:B)_\rho\leq\min_{\substack{\Lambda\in\eff(A)\\\omega_\Lambda\approx_{\eps_{1}}\omega}}\Big\{-\log\lambda_{\min}\left(\Lambda^A\omega^A\Lambda^A\right)\Big\}-H_{\min}^{\eps_0}(A|B)_{\rho}
\ee
where the minimization is over all effects $\Lambda\in\eff(A)$ such that $\omega_\Lambda^{AB}$ as defined in~\eqref{ol} is $\eps_1$-close to $\omega^{AB}$.

\noindent{\it Construction of $\Lambda$:} To upper bound the first term on the right-hand side above, we choose a specific $\Lambda\in\eff(A)$. For this purpose, 
fix $\alpha\in(0,1)$, $\delta\in(0,1/2)$, and let $\tau\in\mb^\delta(A)$ be such that $H_{\alpha}^\delta(A)_\omega=H_{\alpha}(A)_\tau$. From Corollary~\ref{coru} it follows that without loss of generality we can assume that $\tau$ and $\omega$ commutes. Setting $d\eqdef|A|$, we therefore assume that $\omega$ and $\tau$ are diagonal in the same basis and given by
\be
\omega^A=\sum_{x\in[d]}q_x|x\lr x|^A\quad\text{and}\quad\tau^A=\sum_{x\in[d]}t_x|x\lr x|^A
\ee
where both $\q\eqdef(q_1,\ldots,q_d)^T$ and $\t\eqdef(t_1,\ldots,t_d)^T$ are probability vectors in $\prob^\da(d)$; i.e. $\q=\q^\da$ and $\t=\t^\da$. It will be convenient to denote also by $\s\eqdef(s_1,\ldots,s_d)^T$, where for each $x\in[d]$, $s_x\eqdef\min\{q_x,t_x\}$. By definition, $\s=\s^\da$ and
\ba\label{i158}
\|\s\|_1=\sum_{x\in[d]}\min\{q_x,t_x\}&=1-\frac12\left\|\omega^A-\tau^A\right\|_1\\
&\geq 1-\delta\;.
\ea
Since $\delta< 1/2$ we get in particular that $\|\s\|_1> \delta$. We use this property to define $m\in[d]$ as the integer satisfying
\be\label{omin2}
\sum_{x=m+1}^ds_x\leq\delta<\sum_{x=m}^ds_x \;.
\ee
Finally, we take $\Lambda\in\eff(A)$ to be
\be
\Lambda^A=\sum_{x\in[m]}\sqrt{\frac{s_x}{q_x}}|x\lr x|^A\;.
\ee

\noindent{\it Properties of $\Lambda$:}
To get some intuition about this construction observe first that
\ba
\tr\left[\Lambda^A\omega^A\Lambda^A\right]=\sum_{x\in[m]}s_x&=\|\s\|_1-\sum_{x=m+1}^{d}s_x\\
\GG{\eqref{i158},\eqref{omin2}}&\geq 1-2\delta\;.
\ea
Hence, from the gentle measurement lemma we get that $\omega_\Lambda$ is $\sqrt{2\delta}$-close to $\omega$. Note that as long as $2\delta\leq\eps_1^2$ we have that $\omega_\Lambda$ is $\eps_1$-close to $\omega$.

Next, from its definition,
\be
-\log\lambda_{\min}\left(\Lambda^A\omega^A\Lambda^A\right)=-\log s_m\;.
\ee
To relate the right-hand side to an entropic function, observe first
that 
\ba\label{firstt}
H_\alpha^\delta(A)_\omega=H_\alpha(A)_\tau&\geq\frac1{1-\alpha}\log\sum_{x=m}^{d}t_x^\alpha\\
\Gg{t_x\geq s_x}&\geq\frac1{1-\alpha}\log\sum_{x=m}^ds_x^\alpha\;,
\ea
where we restricted the summations to $x\geq m$.
Since for all $x\geq m$ we have $s_m\geq s_x$, we get that
\be
s_x^\alpha=\frac{s_x}{s_x^{1-\alpha}}\geq \frac{s_x}{s_m^{1-\alpha}}\;.
\ee
Substituting this inequality into~\eqref{firstt} we obtain
\ba
H_\alpha^\delta(A)_\omega&\geq-\log s_m+\frac1{1-\alpha}\log\sum_{x=m}^ds_x\\
\GG{\eqref{omin2}}&> -\log s_m+\frac1{1-\alpha}\log\delta\;.
\ea
We therefore arrive at
\ba
-\log\lambda_{\min}\left(\Lambda^A\omega^A\Lambda^A\right)\leq  H_{\alpha}^\delta(A)_{\omega}+\frac1{1-\alpha}\log(1/\delta)\;.
\ea
Combining this with~\eqref{from0} we conclude that as long as $2\delta\leq\eps^2_1$ we have
\be
I_{\max}^\eps(A:B)_\rho\leq H_{\alpha}^\delta(A)_{\omega}-H_{\min}^{\eps_0}(A|B)_{\rho}+\frac1{1-\alpha}\log(1/\delta)\;.
\ee
Next we would like to replace $H_{\alpha}^\delta(A)_{\omega}$ with $H_{\alpha}(A)_{\rho}$. This is possible if $\rho\in\mb^{\delta}(\omega^A)$. Since $\omega^A$ is $\eps_0$-close to $\rho$ the smallest $\delta$ that satisfies  $\rho\in\mb^{\delta}(\omega^A)$ is $\delta=\eps_0$. Since we also want $2\delta\leq\eps^2_1$, we take $\eps_0=\eps_1^2/2$. Finally, the relation $\eps_0+\eps_1\leq\eps$ holds with equality if we take  $\eps_1=(\sqrt{3}-1)\eps$ and $\delta=\eps_0=(2-\sqrt{3})\eps^2$. Denoting by $c\eqdef2-\sqrt{3}$ we conclude that 
for all $\eps\in(0,1)$
 \be
I_{\max}^\eps(A:B)_\rho\leq H_\alpha(A)_{\rho}-H_{\min}^{c\eps^2}(A|B)_{\rho}+\frac1{1-\alpha}\log\left(\frac 1{c\eps^2}\right)\;.
\ee
Thus,
\ba\label{d27}
I_{\max}^\eps(A:B)_\rho&<H_\alpha(A)_{\rho}-H_{\min}^{c\eps^2}(A|B)_{\rho}+\frac1{1-\alpha}\log\left(\frac 1{c\eps^2}\right)\;.\\
\GG{\eqref{direct}}&\leq H_\alpha(A)_{\rho}-\tH_{\beta}^{\ua}(A|B)_\rho+\frac1{\beta-1}\log\left(\frac1{c^2\eps^4}\right)+\frac1{1-\alpha}\log\left(\frac 1{c\eps^2}\right)
\ea
This completes the proof.
\end{proof}

\section{Application: The Reverse Quantum Shannon Theorem}

In the framework of resource theory, the ability to transmit classical information via a quantum channel is understood as the capacity to extract a noiseless (i.e., identity) classical channel between Alice and Bob. In the classical setting, the identity channel corresponds to the completely dephasing channel with respect to the classical basis. Thus, in this context, the “golden unit” of the theory is the completely dephasing channel, denoted as $\Delta_m \in \cptp(Z_A \to Z_B)$, where $Z_A$ and $Z_B$ are classical systems on Alice's and Bob's sides, respectively, with $|Z_A| = |Z_B| = m$. This channel satisfies the defining property of a golden unit~\cite{CG2019,Gour2025}:
\be
\Delta_m\otimes\Delta_n\cong\Delta_{mn}\quad\quad\forall\;m,n\in\mbb{N}\;.
\ee

In the problem of classical communication over a quantum channel, the primary goal is to ``distill", that is, simulate $\Delta_m$ (within a given accuracy) with the largest possible $m$, using only local operations on the channel $\mN$. The reverse quantum Shannon theorem, which we examine here, addresses the complementary problem: simulating the channel $\mN$ using $\Delta_m$ with the smallest possible $m$. To achieve this, we consider a broader set of free operations. Specifically, we assume that shared entanglement is free, while any communication (classical or quantum) is not allowed.

Due to the quantum teleportation protocol, for sufficiently large $m \in \mbb{N}$, any channel can be simulated by LOSE assisted with a classical communication channel $\Delta_m$. Specifically, if $m \geq |A|^2$, then for every channel $\mN \in \cptp(A \to B)$, there exists an LOSE superchannel $\Theta$ such that
$\Theta[\Delta_m]=\mN$.
However, here, for a given channel $\mN^{A \to B}$, we seek the smallest $m$ such that $\Theta[\Delta_m]$ approximates $\mN$ to within a desired accuracy.

By leveraging the superdense coding protocol, the classical golden unit resource $\Delta_m$ with $m = \ell^2$ for some $\ell \in \mbb{N}$ is equivalent to the quantum golden unit $\id_\ell$, which represents an $\ell$-dimensional identity channel from Alice to Bob. Thus, instead of working with the classical golden unit $\Delta_m$, we work with the quantum unit $\id_\ell$. We therefore define the single-shot simulation cost of $\mN$ as
\be\label{78}
\cost^\eps(\mN)\eqdef\min_{\ell\in\mbb{N}}\left\{2\log \ell\;:\;T\left(\id_\ell\xrightarrow{\lose}\mN\right)\leq \eps\right\}
\ee	
where the conversion distance is defined as
\be
T\left(\id_\ell\xrightarrow{\lose}\mN\right)\eqdef\min_{\Theta\in\lose}\frac12\left\|\mN-\Theta\left[\id_\ell\right]\right\|_\diamond\;.
\ee

The reverse quantum Shannon theorem states that for every $\eps\in(0,1)$
\be\label{rh}
\lim_{n\to\infty}\frac1n\cost^\eps\left(\mN^{\otimes n}\right)=I(A:B)_\mN\;.
\ee
The proof of this result is obtained by finding upper and lower bound for $\frac1n\cost^\eps\left(\mN^{\otimes n}\right)$ that converge to to the right-hand side of~\eqref{rh} in the asymptotic limit. Here we focus on finding an upper bound for $\cost^\eps\left(\mN^{\otimes n}\right)$, following ideas presented in~\cite{BCR2011} incorporating the universal upper bound we found in Theorem~\ref{uab}.

In the following theorem, for every $\alpha\in(0,1)$, $\beta>1$, and $n\in\mbb{N}$, we denote by
\be
\delta_n\eqdef\frac1nf_{\alpha,\beta}\left(\frac{\eps}{2(n+1)^{d^2-1}}\right)+\frac{4(d^2-1)\log(n+1)}{n}\;,
\ee
where $f_{\alpha,\beta}(\eps)$ is defined in~\eqref{alep5}, and by
\be
I_{\alpha,\beta}(A:B)_\mN\eqdef\max_{\phi}\Big\{H_{\alpha}(B)_{\omega_\phi}-\tH_{\beta}^\ua(B|R)_{\omega_\phi}\Big\}
\ee
where the maximization is over all $\phi\in\pure(RA)$ with $\omega_\phi^{RB}\eqdef\mN^{A\to B}(\phi^{RA})$.

\subsubsection{Proof of Theorem~\ref{thmj}}\label{reverse}

\noindent\textbf{Theorem.}
{\it For every $\mN\in\cptp(A\to B)$ , $\alpha,\eps\in(0,1)$ and $\beta>1$ the following relation holds:
\be
\frac1n\cost^\eps\left(\mN^{\otimes n}\right)\leq I_{\alpha,\beta}(A:B)_\mN+\delta_n
\ee
where
\be
\delta_n=\frac1nf_{\alpha,\beta}\left(\frac{\eps}{2(n+1)^{d^2-1}}\right)+\frac{4(d^2-1)\log(n+1)}{n}\;.
\ee}

\begin{proof}
To provide an upper bound for $\cost^\eps\left(\mN^{\otimes n}\right)$, 
we make use of
the post-selection technique~\cite{CKR2009} to replace the diamond norm that appears in the conversion distance
\ba
T&\left(\id_\ell\xrightarrow{\lose}\mN^{\otimes n}\right)\eqdef\min
\frac12\left\|\Theta_n\left[\id_\ell\right]-\mN^{\otimes n}\right\|_{\diamond}\;,
\ea
where the minimum is over all LOSE superchannels $\Theta_n$.
Denoting by $\mP_{n,\ell}\eqdef \Theta_n\left[\id_\ell\right]\in\cptp(A^n\to B^n)$ we can assume without loss of generality that $\mP_{n,\ell}$ is $\G$-covariant with respect to the group $\G$ of permutation of $n$ elements, since otherwise we can apply the local twirling superchannel $\Upsilon_n$ such that $\Upsilon_n[\mP_{n,\ell}]$ is covariant and use DPI to show that $\Upsilon_n[\mP_{n,\ell}]$ can only be closer to $\mN^{\otimes n}$ than $\mP_{n,\ell}$ itself.
Thus, from the post selection technique, there exists a purification of a de Finetti state $\xi_n\in\pure(C_nA^n\tA^n)$, with $|\tA|=|A|$ and $|C_n|\leq (n+1)^{|A|^2-1}$, such that $\mP_{n,\ell}\approx_\eps\mN^{\otimes n}$ whenever
\be\label{82}
\mP_{n,\ell}^{\tA^n\to B^n}\left(\xi_n^{C_nA^n\tA^n}\right)\approx_{\eps_n}\left(\mN^{\tA\to B}\right)^{\otimes n}\left(\xi_n^{C_nA^n\tA^n}\right)\;,
\ee
where $\eps_n\eqdef\frac12\eps(n+1)^{-(|A|^2-1)}$. Recall that $\tA$ is a replica of $A$ and observe that w.l.o.g.\ we switched the role of $A$ and $\tA$ (this will simplify later some expressions).

A key property of the pure de Finetti state $\xi_n^{C_nA^n\tA^n}$ is that its marginal $\xi^{A^n\tA^n}$ can be written as a convex combination of states of the form $\phi^{\otimes n}$ with $\phi\in\pure(A\tA)$. 
Applying this to the cost of $\mN^{\otimes n}$ as given in~\eqref{78} we arrive at
\be
\cost^\eps\left(\mN^{\otimes n}\right)\leq \min_{\ell\in\mbb{N}}\left\{2\log \ell\;:\;\mP_{n,\ell}(\xi_n)\approx_{\eps_n}\mN^{\otimes n}(\xi_n)\right\}
\ee
where the minimum is also over all LOSE superchannels $\Theta_n$
such that $\mP_{n,\ell}\eqdef \Theta_n\left[\id_\ell\right]\in\cptp(\tA^n\to B^n)$.

Next, we construct the following three-step protocol. Let $A'$ be a replica of $B$ on Alice side, with $\mN^{\tA\to A'}\eqdef\id^{B\to A'}\circ\mN^{\tA\to B}$ being a channel on Alice's side, and let $\mV\in\cptp(\tA\to EA')$with $|E|\leq |\tA  A'|$ be its Stinespring isometry channel. We assume here that system $E$ is on Alice's side. The protocol consists of three key steps:
\ben
\item  Alice simulates $\mV^{\otimes n}$ in her lab.
\item Alice applies the channel $\mQ_{n,\ell}=\Theta_n'\left[\id_\ell\right]\in\cptp(E^nA'^n\to E^nB^n)$, where the superchannel $\Theta_n'$ is an $\eps_n$-error QSS protocol that is used to simulate the action of the channel $\id^{E^nA'^n\to E^nB^n}$ on the pure state 
\be
\rho^{C_nA
^nE^n{A'}^n}\eqdef\mV^{\otimes n}\left(\xi_n^{C_nA^n\tA^n}\right)\;.
\ee
\item Alice discard system $E^n$.
\een
We now discuss the technical details of the protocol.

By construction, the three steps of the protocol result with the channel
\be\label{pn}
\mP_{n,\ell}^{\tA^n\to B^n}=\tr_{E^n}\circ\mQ_{n,\ell}^{E^n{A'}^n\to E^nB^n}\circ \left(\mV^{\tA\to EA'}\right)^{\otimes n}\;.
\ee 
Since $\Theta'_n$ is  $\eps_n$-QSS with respect to the state $\rho^{A^nC_nE^n{A'}^n}$. Thus,
\be\label{a0}
\mQ_{n,\ell}^{E^nA'^n\to E^nB^n}\left(\rho^{A^nC_nE^n{A'}^n}\right)\approx_{\eps_n}\rho^{A^nC_nE^nB^n}
\ee
and
\be
\log\ell\leq I_{2}^{\eps_n-\delta_n}(A^nC_n:{A'}^n)_{\rho}+\log\left(\frac1{\delta_n}\right)\;,
\ee
for every $\delta_n\in(0,\eps_n)$.
By tracing out system $E^n$ in~\eqref{a0} we get from the DPI that~\eqref{82} holds with  the channel defined in~\eqref{pn}. Thus, 
$\mP_{n,\ell}\approx_\eps\mN^{\otimes n}$ so that we arrive at the upper bound
\be
\cost^{\eps}\left(\mN^{\otimes n}\right)\leq I_{2}^{\eps_n-\delta_n}(A^nC_n:{A'}^n)_{\rho}+\log\left(\frac1{\delta_n}\right)\;.
\ee
This upper bound is tighter than the one given in~\cite{BCR2011} since it depends on $I_2$ rather than $I_{\max}$. However, since we are interested here in its asymptotic behaviour, particularly in its dependance on $n$, we will choose $\delta_n=\frac12\eps_n$ and replace $I_2$ with $I_{\max}$. In addition, we use the property proved in~\cite{BCR2011} that
\begin{align}
I_{\max}^{\eps_n/2}(A^nC_n:{A'}^n)_{\rho}
&\leq I_{\max}^{\eps_n/2}(A^n:{A'}^n)_{\rho}+2\log|C_n|\nonumber\\
&\leq I_{\max}^{\eps_n/2}(A^n:{A'}^n)_{\rho}+\nu_n\;,
\end{align}
where $\nu_n\eqdef2(|A|^2-1)\log(n+1)$.

To simplify further this expression, recall that $\rho^{A^n{A'}^n}=(\mN^{\tA\to A'})^{\otimes n}\left(\xi^{A^n\tA^n}\right)$, where $\xi^{A^n\tA^n}$ is a de-Finettie state and therefore can be expressed as
\be
\xi^{A^n\tA^n}=\sum_{j\in[k]}t_j\phi_{j}^{\otimes n}
\ee
where $k\eqdef (n+1)^{2(|A|^2-1)}$, $(t_1,\ldots,t_k)^T\in\prob(k)$, and for each $j\in[k]$, $\phi_j\in\pure(A\tA)$.  For every $\phi\in\pure(A\tA)$ we denote by $\omega_\phi^{AB}\eqdef\mN^{A\to B}(\phi^{A\tA})$.
Thus, from the type of quasi-convexity of $I_{\max}^\eps$ proven in~\cite{BCR2011} one obtain that (recall $A'\cong B$)
\begin{align}
I_{\max}^{\eps_n}(A^n:B^n)_{\rho}\leq \max_{\phi}I_{\max}^{\eps_n}(A^n:B^n)_{\omega_\phi^{\otimes n}}+\nu_n\;,
\end{align}
where the maximum is over all $\phi\in\pure(A\tA)$.
Combining everything we conclude that for every $\eps\in(0,1)$
\begin{align}
\frac1n\cost^{\eps}\left(\mN^{\otimes n}\right)
&\leq\frac1n\max_{\phi}I_{\max}^{\eps_n}(A^n:B^n)_{\omega_\phi^{\otimes n}}+\frac{2\nu_n}{n}\nonumber\\
\GG{Theorem~\ref{uab}}&\leq \max_{\phi}\Big\{H_{\alpha}(A)_{\omega_\phi}-\tH_{\beta}^\ua(A|B)_{\omega_\phi}\Big\}+\mathcal{O}\left(\frac{\log(n)}{n}\right)\;,
\end{align}
where we used the fact that both $\nu_n$ and $\log(\eps_n)$ scale as $\mO\big(\log (n)\big)$. This completes the proof.
\end{proof}



\end{document}